\documentclass[manuscript]{acmart}

\usepackage[utf8]{inputenc}
\usepackage{amsthm,amssymb,amsmath}
\usepackage{xspace}

  \theoremstyle{plain}
  \newtheorem{theorem}{Theorem}
  \newtheorem{lemma}[theorem]{Lemma}
  \newtheorem{corollary}[theorem]{Corollary}

  \newtheorem{definition}[theorem]{Definition}
  
  \newtheorem{remark}[theorem]{Remark}
  \newtheorem{claim}[theorem]{Claim}

\usepackage[ruled]{algorithm}
\usepackage{algorithmic}

\usepackage{graphicx}
\usepackage{microtype}
\usepackage{xspace}
\usepackage{verbatim}
\usepackage{units}
\usepackage{color}
\usepackage[title]{appendix}
\usepackage{enumitem}
\usepackage{subcaption}

\AtBeginDocument{%
  \providecommand\BibTeX{{%
    \normalfont B\kern-0.5em{\scshape i\kern-0.25em b}\kern-0.8em\TeX}}}

\copyrightyear{2022} 
\acmYear{2022} 
\setcopyright{rightsretained} 
\acmConference[EAAMO '22]{Equity and Access in Algorithms, Mechanisms, and Optimization}{October 6--9, 2022}{Arlington, VA, USA}
\acmBooktitle{Equity and Access in Algorithms, Mechanisms, and Optimization (EAAMO '22), October 6--9, 2022, Arlington, VA, USA}
\acmDOI{10.1145/3551624.3555293}
\acmISBN{978-1-4503-9477-2/22/10}

\begin{document}

\title{A Heterogeneous Schelling Model for Wealth Disparity and its \\ Effect on Segregation}

\author{Zhanzhan Zhao}
\authornote{Corresponding authors. This work was supported by Georgia Tech ARC-TRIAD Fellowship, NSF Award CCF-2106687, CCF-1733812, and ARO MURI Award W911NF-19-1-0233.}
\email{zhanzhanzhao@gatech.edu}
\affiliation{%
 \institution{Georgia Institute of Technology}
 \department{School of Computer Science}
 \country{USA}}

\author{Dana Randall}
\authornotemark[1]
\email{randall@cc.gatech.edu}
\affiliation{%
 \institution{Georgia Institute of Technology}
 \department{School of Computer Science}
 \country{USA}}

\begin{abstract}
 The Schelling model of segregation was introduced in economics to show how micro-motives can influence macro-behavior. 
 Agents on a lattice have two colors and try to move to a different location if the number of their neighbors with a different color exceeds some threshold. Simulations reveal that even such mild local color preferences, or homophily, are sufficient to cause segregation.
 In this work, we propose a stochastic generalization of the Schelling model, based on both race and wealth, to understand how carefully architected placement of incentives, such as urban infrastructure, might affect segregation. 
 In our model, each agent is assigned one of two colors along with a label, rich or poor. Further, we designate certain vertices on the lattice as “urban sites,” providing civic infrastructure that most benefits the poorer population, thus incentivizing the occupation of such vertices by poor agents of either color. 
 We look at the stationary distribution of a Markov process reflecting these preferences to understand the long-term effects. 

 We prove that when incentives are large enough, we will have "urbanization of poverty," an observed effect whereby poor people tend to congregate on urban sites. Moreover, even when homophily preferences are very small, if the incentives are large and there is income inequality in the two-color classes, we can get racial segregation on urban sites but integration on non-urban sites. In contrast, we find an overall mitigation of segregation when the urban sites are distributed throughout the lattice and the incentives for urban sites exceed the homophily biases. We prove that in this case, no matter how strong homophily preferences are, it will be exponentially unlikely that a configuration chosen from stationarity will have large, homogeneous clusters of agents of either color, suggesting we will have racial integration with high probability.
\end{abstract}



\maketitle

\section{Introduction}
\label{introduction}
Over fifty years ago,  economist Thomas Schelling studied segregation by modeling residents as colored particles on a chessboard. 
Each particle is considered {\it happy} if its color agrees with more than a fixed fraction of its neighbors and {\it unhappy} particles try to move to new locations with more favorable neighborhoods \cite{hatna2014combining}. 
Simulations reveal even a mild preference for neighbors of one's own color is sufficient to cause segregation on a macroscopic scale~\cite{schelling1971dynamic}. %
Extensive work has been done by economists and sociologists to expand Schelling's model using statistical analysis, simulation tools, and enhanced models \cite{bayer2012tiebout, kortum2012joint, lehman2011segregation, tammaru2020relationship}. 
This work primarily focuses on how the dynamics determine the limiting distribution and try to connect the model to the real world population dynamics \cite{clark1991residential, clark2008understanding, laurie2003role, singh2009schelling,cottrell2017multidimensional, perez2019geospatial,yinger1976racial}. Recent work also seeks to understand the dual segregation of ethnicity and wealth with empirical studies \cite{sahasranaman2018ethnicity, fossett2011generative, paolillo2018different}.

Additional heuristical and rigorous studies on the implications of Schelling-like dynamics have been undertaken in the theoretical computer science and statistical physics communities, where the concept of micro-motives affecting macro-behavior such as phase transitions is well-understood. For instance, Brandt et al. \cite{brandt2012analysis, gerhold2008computing} rigorously determined the precise limiting distributions for the Schelling model in one dimension.  Additional rigorous analysis was provided for a modified Schelling model with simplified neighborhood interactions~\cite{barmpalias2014digital, stauffer2007ising, pollicott2001dynamics} or with generalized local interactions \cite{stauffer2008social, bhakta2014clustering}. Bhakta et al. \cite{bhakta2014clustering} introduced a randomized variant and proved that slight biases maintain well-integrated populations, whereas stronger biases lead to segregation. 
Unlike Schelling's model where each person's happiness has a deterministic threshold regarding one's tolerance for differently colored neighbors, the model in \cite{bhakta2014clustering} allows all particles to move stochastically and they are increasingly inclined to move when they have more neighbors of the opposite color.   
Improved bounds on the amount of bias that leads to integration and segregation were given by Cannon et al. \cite{cannon2019local} for specific geometric incentive functions where the model can be mapped onto problems of heterogeneous particle separation in the programmable matter.

Most variants of the Schelling model assume that agents of each race are homogeneous and have identical incentives influencing where they prefer to live purely based on {\it homophily}, the desire for each particle to have neighbors that are similar to oneself, regardless of socio-economic status and location. However, such simple models cannot explain two widely observed phenomena: {\it centralization}, whereby one racial group clusters near the the city center, and {\it urbanization of poverty}, whereby city centers  and other areas dense with public amenities and infrastructure disproportionately attract the poorer populations.  Centralization  is widely-used to measuring racial segregation in metropolitan areas \cite{massey1988dimensions, iceland2002racial}. Urban economists show that urbanization of poverty results from better access to public transportation in central cities and other resources~\cite{glaeser2008poor}.  Such evidence shows us that socio-economic considerations such as the spatial distributions of urban infrastructure are significant factors influencing racial segregation but these are not captured by any of the theoretical models. This motivated our work which simultaneously considers both homophily and each individuals' incentives according to their wealth level and their access to public amentities.  With our proposed new model, we are able to rigorously explore the impact of wealth disparity on racial segregation, as well as civic interventions to potentially help mitigate segregation.

\subsection{The heterogeneous Schelling model}
To better understand these socioeconomic distinctions and the effects of economic disparity within a city, we introduce a new 
{\it heterogeneous Schelling model} where  individuals are each assigned  a color and designated  {\it rich} or {\it poor}. We also distinguish some vertices on the underlying lattice to be {\it urban sites} if they  provide useful infrastructure (or resources) that is most beneficial to  poor citizens.  The urban sites might be  grouped centrally, for instance representing a metropolitan city center, or  distributed evenly throughout large parts of the city, representing a vast public transportation network or other distributed amenities (see Fig~\ref{1}). While all individuals have uniform homophily preferences, as in the standard Schelling model, we add additional incentives that favor configurations with more poor people residing on urban sites, capturing the presumption that urban sites provide sufficient benefits to poor individuals to incentivize their relaxing their racial biases. {We are interested in understanding when urban infrastructure can help mitigate racial biases and lessen segregation for various placements of urban sites for such a model.}

Specifically, we  represented the city by a finite torus on the triangular lattice, with each site accommodating exactly one person. Each person (or agent)  is  blue or red, representing race, and  {\it rich} or {\it poor}, representing wealth.  The vertices $\mathcal{U} \subseteq \mathcal{V}$ are the {\it urban sites}.   Each pair of neighbors has homophily (or racial) bias $\lambda$, representing how much they each prefer neighbors of their own color. Setting $\lambda > 1$ is the ``ferromagnetic'' setting corresponding to agents preferring same-colored neighbors.
Further, poor agents have an affinity for urban sites with a wealth bias parameter $\gamma$;
setting $\gamma > 1$ biases poor agents to prefer residing on urban sites. When $\gamma = 1,$ we recover the pure standard homophily model where wealth of individuals is not considered. 
Let $\Omega = (\{{\rm red}, {\rm blue}\} \times \{{\rm rich}, {\rm poor}\})^{|V|}$ be the state space.  The stationary probability of any configuration $\sigma \in  \Omega$ is given by $$\pi(\sigma) = \lambda^{-h(\sigma)}\gamma^{p(\sigma)}/Z,$$
where $h(\sigma)$ is the number of racially heterogeneous edges (whose endpoints do not share the same color), $p(\sigma)$ is the number of poor agents on urban sites, and 
$$Z = \sum_{\sigma \in \Omega} \lambda^{-h(\sigma)}\gamma^{p(\sigma)}$$
is the normalizing constant.

A randomized algorithm $\mathcal{M}$ for sampling from $\pi$ can be described as follows.
At each time step, two random agents are selected, and they swap locations with the appropriate Metropolis probabilities  so as to converge to $\pi$.   In particular, they are more likely to swap  if they are each in less homogeneous neighborhoods, as previously studied in \cite{bhakta2014clustering, cannon2019local}, with an additional bias toward keeping poor agents on urban sites, so happier individuals are less likely to move.
We note that when there are no urban sites (or all vertices are urban sites), then the wealth of individuals becomes irrelevant and we recover the racial segregation model studied in \cite{cannon2019local}, where the dichotomy of the phase change between integration and segregation has been proved. Here we are interested in the effects in heterogeneous cases where both urban and non-urban sites are present. We also require the size of the urban sites to be of a constant fraction of the total sites. For topology, we study the impact of the centralized or distributed placement of the urban sites on segregation.

\subsection{Effects on wealth and racial segregation}
First, we show that our model yields {\it urbanization of poverty} when the wealth bias $\gamma$ is sufficiently large, with all but an arbitrarily small fraction of urban sites being occupied by poor agents.  
Conversely, we show that for any racial bias $\lambda>1$, if the wealth bias $\gamma>1$ is small enough, then it is exponentially unlikely that poor agents will be disproportionately concentrated on urban sites. 

Moreover, when the urban sites are centralized and  both racial bias $\lambda$ and wealth bias $\gamma$ are large enough, {\it urbanization of poverty} and {\it racial segregation} will occur simultaneously. However, when there is significant inequality in the distribution of wealth and many more poor people come from one race, then
even when the racial bias $\lambda$ is small, as long as the wealth bias $\gamma$ is large, we will have {\it racial segregation on urban sites} and {\it racial integration on the non-urban sites}. This suggests that the urbanization of poverty can enhance racial segregation when the infrastructure is centralized, such as with a dense city center with civic services and perhaps subsidized housing, providing a primary location that incentivizes occupation by poor people.

We show there will be a dramatically different outcome when the urban sites are well-distributed throughout the city, such as with public transportation stops that service the entire city. First, we prove under income inequality, where one race has a higher proportion of poor people, no matter how large racial bias $\lambda$ is, as long as the wealth bias $\gamma$ exceeds racial bias $\lambda$ sufficiently, both the urban and non-urban sites will be integrated with high probability. That is, the probability of large spatial clusters with predominantly one race forming anywhere is exponentially small. This suggests that distributing urban infrastructure equitably throughout the city will have a better effect on mitigating segregation when the incentives are large enough compared to the inherent racial biases.

Our proofs build on {\it Peierls arguments} from statistical physics for the integration and separation of heterogenous particles in the context of programmable matter \cite{cannon2019local}. The essential idea is to map the set of configurations not satisfying a target property to a set of configurations that have exponentially larger probability at stationarity, so that the inverse maps do not require significant information, thus proving that configurations outside of the target set must have small probability by  evaluating ``energy\slash entropy'' balancing the probabilities and the number of preimages. However, the introduction of urban sites and wealth bias greatly complicates the proofs as we have to keep the same number of people for each pair of wealth level and race before and after the mapping $\nu = f(\sigma).$  In our setting, all four groups may deviate under the maps and it requires careful arguments to be able to  restore the cardinalities of all the sets without losing too much information about the inverse map, which is significantly more challenging than earlier proofs that only considered race.

\section{Preliminaries}
The dynamics we study can be viewed as a Markov chain that converges to a distribution reflecting the overall effects of individual biases.  We briefly review properties of Markov chains and summarize techniques used to analyze their stationary distributions.

\subsection{Markov chains}\label{A0}
A Markov chain is a memoryless random process on a state space $\Omega$, which is is finite and discrete in our setting. We focus on discrete time Markov chains, where one transition occurs per time step. The transition matrix $M$ on $\Omega \times \Omega \to [0, 1]$ is defined so that $M(x, y)$ is the probability of transiting from state $x$ to state $y$ in one step, for any pair $x, y \in \Omega$. The $t-$step transition probability $M^{t}(x, y)$ is the probability of moving from $x$ to $y$ in exactly $t$ steps.

A Markov chain is {\it irreducible} if for all $x, y \in \Omega$, there exists a $t$ such that $M^{t}(x, y) > 0$ and is ${\it aperiodic}$ if for all $x, y \in \Omega,$ g.c.d.$\{t: M^{t}(x, y) > 0\} = 1,$ where g.c.d. stands for the greatest common divisor. A Markov chain is {\it ergodic} if it is {\it irreducible} and {\it aperiodic} (see, e.g.,  \cite{levin2017markov}). 

A {\it stationary distribution} of a Markov chain is a probability distribution $\pi$ over the state space $\Omega$ such that $\pi P = \pi.$ Any finite, ergodic Markov chain converges to a unique stationary distribution given by $\pi(y) = \lim_{t \to \infty}P^{t}(x, y)$ for any $x, y \in \Omega;$ moreover, for such chains the stationary distribution $\pi(y)$ is independent of starting state $x.$
To verify $\pi^{'}$ is the unique stationary distribution of a finite ergodic Markov chain, it suffices to check the {\it detailed balance condition}, i.e., $\pi^{'}(x)M(x, y) = \pi^{'}(y)M(y, x)$ and $\sum_{x \in \Omega} \pi^{'}(x) = 1$ for all $x, y \in \Omega$ (see \cite{feller1957introduction}).

\subsection{Peierls arguments} \label{PA}
Peierls arguments are helpful in analyzing a chain's limiting behavior by showing that the probability a sample drawn from the stationary distribution $\pi$ of a Markov chain falls into some target set is exponentially small in $n$,  indicating that $\pi(\Omega_{t}) \leq \xi^{n}$ for some constant $\xi \in (0, 1).$
The Peierls argument is based on a map from configurations in the target set $\Omega_{\rm t}$ to configurations in the configuration space $\Omega$ such that the map has an exponential gain in probability. Thus the targeted configurations are exponentially unlikely in $\Omega$.
Mathematically, the mapping  $\nu = f(\sigma)$ is defined from $\sigma \in \Omega_{\rm t}$ to $\nu \in \Omega$, which yields
\begin{align}\label{map1}
	\pi(\Omega_{\rm t}) = \sum_{\sigma \in \Omega_{\rm t}} \pi(\sigma) &\leq \sum_{\nu \in \Omega} \sum_{\sigma \in f^{-1}(\nu)} \pi(\sigma)\nonumber\\ &= \sum_{\nu \in \Omega} \pi(\nu) \frac{\sum_{\sigma \in f^{-1}(\nu)} \pi(\sigma)}{\pi(\nu)}.
\end{align} 
In order to show $\pi(\Omega_{\rm t}) \leq \xi^{n}$, the mapping needs to be carefully defined to get the upper bounds of the probability ratio $\frac{\pi(\sigma)}{\pi(\nu)}$ and the number of the preimages $|f^{-1}(\nu)|$ for any given $\nu$.
A mapped configuration $\nu$ with large probability ratio $\frac{\pi(\sigma)}{\pi(\nu)}$ can also have many possible preimages, which necessitates carefully balancing  $\frac{\pi(\sigma)}{\pi(\nu)}$ and $|f^{-1}(\nu)|$,  representing an {\it energy\slash entropy tradeoff.} 

To facilitate the mapping operations $f(\cdot)$ and the counting of $|f^{-1}(\nu)|$,  certain {\it bridge systems} have been introduced in \cite{miracle2011clustering, cannon2019local} to efficiently encode some information about mapped configurations to facilitate inverting the map and help bound the number of preimages. 
Unlike \cite{miracle2011clustering, cannon2019local}, here the configuration space is enlarged by the introduction of a wealth dimension,  requiring extending the bridge system to encode multi-dimensional information representing race and wealth. Moreover, because the additional wealth term is reflected in the stationary distribution, more careful mapping rules are required to account for tradeoffs  between  $\frac{\pi(\sigma)}{\pi(\nu)}$ and $|f^{-1}(\nu)|,$ balancing the effects of both the wealth and homophily biases in the probability measure.

\section{The Heterogeneous Schelling Model with Incentives}
\label{urb_seg_alg}
In our proposed model, a city is represented by a finite toroidal region of the triangular lattice $G_{\triangle} = (\mathcal{V}, \mathcal{E})$, shown in Figure~\ref{1a}. Each vertex in $\mathcal{V}$ represents a potential residence or {\it site}. Two adjacent vertices are {\it neighboring sites}, and each site has six nearest neighbors on $G_{\triangle}$. Some  vertices $\mathcal{U} \subseteq \mathcal{V}$ are designated  {\it urban sites}.   We denote the set of agents as $\mathcal{A}$  and the poor agents as $\mathcal{P \subseteq A}.$   Figure~\ref{1a} shows an example of centralized placement of urban sites, whereas Figure~\ref{1b} shows the distributed placement, with urban sites depicted as yellow hexagons.

We assume each agent $i$ is assigned a race $r(i) \in \{\text{blue, red}\}$ and  wealth $w(i) \in \{\text{rich, poor}\}.$ Each site in $\mathcal{V}$ can accommodate at most one agent. For simplicity of analysis, we assume that $n$ agents fully occupy all the sites on $G_{\triangle}$, where $|\mathcal{V}| = n.$ The size of the urban sites is assumed to be of a constant fraction of all the sites, i.e., $|\mathcal{U}| = c \cdot n,$ where $c \in [0, 1].$ 
As shown in Figure~\ref{1c},  we represent the race of an agent by color and the wealth of an agent by the shade of each color; poor blue agents are referred to as  {\it cyan}, poor red agents are {\it pink} and {\it blue} and {\it red} are reserved for the rich members of each color class.

Among the $n$ agents, $\mathcal{P}$ is the subset that are poor; the fraction that are poor is denoted by $p$, so  $|\mathcal{P}| = p \cdot n.$ Similarly, the fraction of red agents $\mathcal{R}$ is $r$, with $|\mathcal{R}| = r \cdot n.$  Among the red agents, we further denote the fraction of poor red agents $\mathcal{R}_{\rm p}$ as $r_{\rm p}$, and the fraction of  rich red agents $\mathcal{R}_{\rm r}$ as $r_{\rm r}$, so that $|\mathcal{R}_{\rm p}| = r_{\rm p} n,$ and $|\mathcal{R}_{\rm r}| = r_{\rm r} n.$ Similarly, we define the fraction of the blue $\mathcal{B}$ as $b$, the fraction of poor blue $\mathcal{B}_{\rm p}$ as $b_{\rm p}$, and the fraction of  rich blue as $\mathcal{B}_{\rm r}$ as $b_{\rm r}$.

\begin{figure*}[t]
	\begin{subfigure}[t]{0.3\textwidth} 
		\centering
		\includegraphics[height=1.22in]{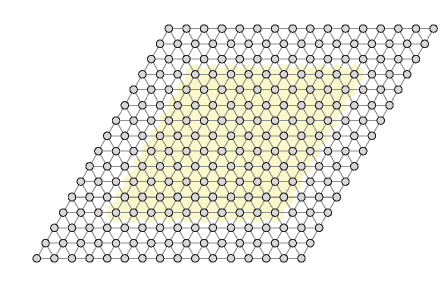}
		\captionsetup{justification=centering}
		\subcaption{} \label{1a}
	\end{subfigure}
	\begin{subfigure}[t]{0.3\textwidth} 
		\centering
		\includegraphics[height=1.22in]{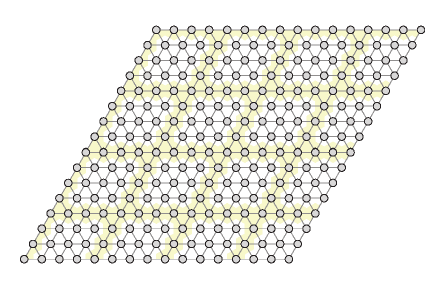}
		\captionsetup{justification=centering}
		\subcaption{}\label{1b}
	\end{subfigure}
	\begin{subfigure}[t]{0.3\textwidth} 
		\centering
		\includegraphics[height=1.22in]{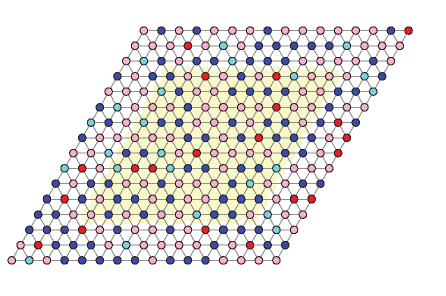}
		\captionsetup{justification=centering}
		\subcaption{}\label{1c}
	\end{subfigure}
	\caption[]{A city lattice region $G_{\triangle}$  (a) with centralized urban sites (shaded in yellow), (b) with distributed urban sites, and (c) centralized and fully occupied by the four types of agents.  }\label{1}
\end{figure*}

A {\it configuration} (or a {\it state}) $\sigma$ is an assignment of the four types (race and wealth) to each of the vertices of $G_{\triangle}$. The {\it state space} (or configuration space) $\Omega = (\{{\rm red}, {\rm blue}\} \times \{{\rm rich},{\rm poor}\})^{|V|}$ is the set of all possible  configurations.

For a configuration $\sigma,$ we denote  $\ell_{\sigma}(i) \in \mathcal{V}$ as the site where agent $i$ resides. Agents living at adjacent sites are {\it neighbors} and each agent has six neighbors.
Each agent $i$ is assigned a race $r(i)$, wealth $w(i)$, and occupies a site $\ell_{\sigma}(i)$, which it can recognize as an urban site  or not. We define an indicator function that takes agent $i$ as input and outputs true when $i$ is poor and currently on the urban sites as the following:
$$u_{\sigma}(i) \triangleq \left\{\begin{matrix}
	1, &\text{if\,\,} i \in \mathcal{P}, \ell_{\sigma}(i) \in \mathcal{U}\\ 
	0,&\text{otherwise}
\end{matrix}\right.$$ 
For a configuration $\sigma,$ the number of agents that are both poor and on the urban sites is defined to be $p(\sigma) \triangleq \sum_{i \in \mathcal{A}}u_{\sigma}(i).$

For each agent $i$, let $N_{\sigma}(i)$ be the number of neighbors of $i$ that share its color. 
An edge in a configuration $\sigma$ with vertices occupied by agents $i$ and $j$ is {\it racially homogeneous} if their colors agree (i.e., $r(i) = r(j)$) and {\it racially heterogeneous} otherwise. We define the total number of racially heterogeneous edges of a configuration $\sigma$ as $h(\sigma),$ and the total number of racially homogeneous edges as $e(\sigma).$

The Markov chain $\mathcal{M}$ is defined so that it will converge to $\pi(\sigma) = \lambda^{-h(\sigma)} \gamma^{p(\sigma)} / Z$, which generalizes the Schelling probabilities to reflect the additional contribution of urban sites.
Each agent $i$ is able to swap locations with any agent $j \in \mathcal{A}, j \neq i$ in the city $G_{\triangle}$, which we denote it a {\it swap move} $s_{ij}.$
 Beginning with any configuration $\sigma_0 \in \Omega,$ at each time step, 
 the algorithm randomly picks two agents $i$ and $j$ at sites $\ell_{\sigma}(i) \in \mathcal{V}$ and $\ell_{\sigma}(j) \in \mathcal{V}$ and tries to swap their positions with the appropriate Metropolis probabilities (so agents are  more likely to move if the move increases its number of racially homogeneous neighbors, but with a  dampening factor $\frac{1}{\gamma} < 1$ if the agent is poor and currently at the urban site. Mathematically, 
$$P(\sigma: i \to j) = \frac{\lambda^{-N_{\sigma}(i)}}{\gamma^{u_{\sigma}(i)}},$$
where $\lambda > 1$, and $\gamma > 1$. The probability of agents $i$ and agent $j$ swapping positions satisfies 
\begin{align}
		P(\sigma: s_{ij}) = \frac{1}{n^2} \lambda^{-N_{\sigma}(i)-N_{\sigma}(j)} \gamma^{-\sum_{k \in \{i,j\}} u_{\sigma}(k)}.\label{transition}
\end{align}

\begin{algorithm}
\renewcommand{\algorithmicrequire}{\textbf{Beginning at any configuration $\sigma_0$ with $n$ agents, repeat: }}
	\caption{Markov Chain $\mathcal{M}$.}\label{alg}
	\begin{algorithmic}
		\REQUIRE {}
		\STATE {Choose two agents $i$ and $j$ uniformly at random in the current configuration $\sigma$.}
		\STATE {Choose $q \in (0, 1)$ uniformly at random.}
		$\,\,\,\,$ \IF {$q < \lambda^{-N_{\sigma}(i)-N_{\sigma}(j)} \gamma^{-\sum_{k \in \{i,j\}} u_{\sigma}(k)}$} 
		    \STATE{agents $i$ and $j$  swap positions.}
		\ELSE 
		    \STATE { agents $i$ and $j$ keep their current locations.}
		\ENDIF
	\end{algorithmic}
\end{algorithm}

It is easy to see that the Markov chain $\mathcal{M}$ is ergodic on the state space $\Omega,$ since  swap moves of $\mathcal{M}$ suffice to transform any configuration to any other configuration (irreducible) and there is a non-zero self-loop probability for $\lambda > 1$ and $\gamma >1$ (aperiodic).
\noindent Using detailed balance it is easy to confirm that the Markov chain converges to 
\begin{align}
			\pi(\sigma) = \lambda^{-h(\sigma)} \gamma^{p(\sigma)} / Z,\label{pi}
\end{align} 
with $h(\sigma)$ the number of racially heterogeneous edges in $\sigma$, $p(\sigma)$ represent the number of poor people on the urban sites in $\sigma$, and $Z = \sum_{\sigma \in \Omega} \lambda^{-h(\sigma)} \gamma^{p(\sigma)}$ the partition function that normalizes the probability distribution. See Section 1.1 of Supplementary Information (SI) for proof details.

\section{Urbanization of Poverty}
\begin{figure*}
		\centering
		\begin{subfigure}[t]{0.45\textwidth}
			\centering
	\includegraphics[width=\textwidth]{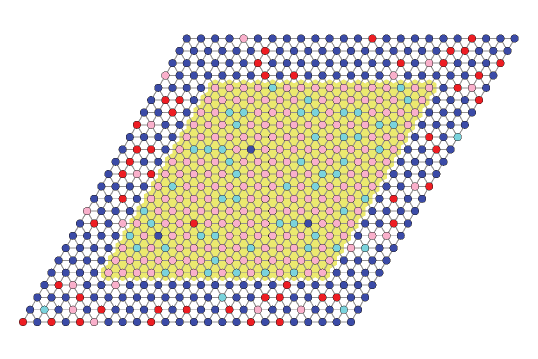}
			\caption[]%
			{{\footnotesize mixed urbanization and racial integration outside the centralized urban sites, with $\gamma = 200, \lambda = 1.01$. }}
			\label{fig:sim1a}
		\end{subfigure}
		\hfill
		\begin{subfigure}[t]{0.45\textwidth}  
			\centering 
	\includegraphics[width=\textwidth]{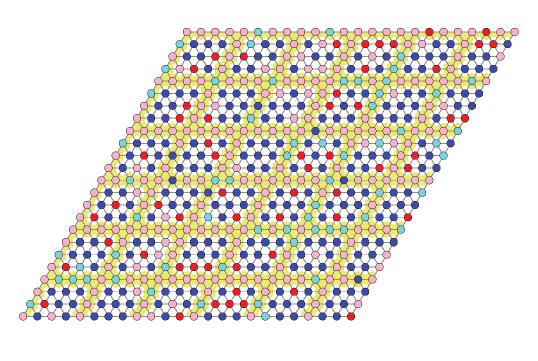}
			\caption[]%
			{{\footnotesize urbanization of poverty and integration with distributed urban sites with $\gamma = 200, \lambda = 1.01$.}} 
		\label{fig:sim1b}
		\end{subfigure}
\caption[]{Simulations of $\mathcal{M}$ after five million iterations with 40$\%$ poor red(pink), 10$\%$ rich red (red), 40$\%$ rich blue (blue), and 10$\%$ poor blue (cyan) and a $46\%$ fraction of the urban sites.}\label{fig:sim1}
\end{figure*} 
We begin by confirming the {\it urbanization of poverty}, whereby a $1-\epsilon$ fraction of the urban sites are occupied by poor agents, for any constant $\epsilon>0$.
We prove in Theorem~\ref{urb1} that under {\it centralized placement}, for any $\lambda$ and $\epsilon$, and $\gamma$ sufficiently large,  urbanization of poverty will occur at stationarity with high probability. See Figure~\ref{fig:sim1a} for simulations.
Further, in Corollary~\ref{urb2} we prove that when urban sites are distributed throughout the lattice, for any $\epsilon$ and $\lambda$ with   $\gamma>\lambda$,  we again will have urbanization of poverty (simulations in Figure~\ref{fig:sim1b}). 
Finally, we prove that for centralized urban sites, if $\gamma > 1$ and $\epsilon>0$ are sufficiently small, then for any $\lambda > 1$, we are very unlikely to have urbanization of poverty, i.e., more than a $1-\epsilon$ fraction of urban sites will be occupied by rich agents.

\begin{definition} \label{def_urb}
	For any $\epsilon \in (0, \frac{1}{2})$, a city is said to have \textbf{$\epsilon-$urbaniza-tion of poverty} if the number of poor agents on the urban sites is at least $\min\{c,p\} n - \epsilon n.$
\end{definition}
\noindent The parameter $\epsilon$ captures the tolerance for allowing rich agents on urban sites: smaller $\epsilon$ indicates a higher density of poor agents on urban sites, whereas larger $\epsilon$ allows a larger fraction to be occupied by rich agents. The term $\min\{c,p\} n$ is the maximum possible occupancy of poor agents on urban sites for given populations.
Hence, $\epsilon-$urbanization of poverty requires that the maximum number of poor agents occupy urban sites, leaving at most an $\epsilon$ fraction to be occupied by rich agents.

First, we show in Theorem \ref{urb1} and Corollary \ref{urb2} that for either centralized urban sites or distributed urban sites, if $\gamma$ is sufficiently large, then we are likely to observe urbanization of poverty at stationarity.

\begin{figure*}
		\begin{subfigure}[t]{0.45\textwidth}   
			\centering 
	\includegraphics[width=\textwidth]{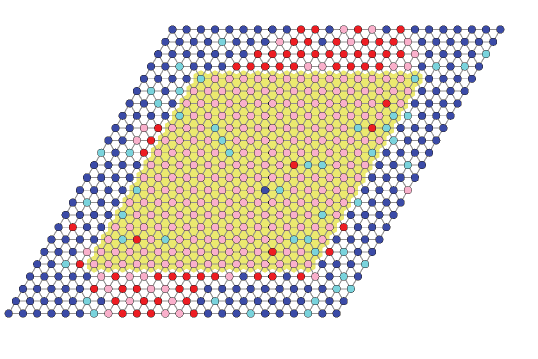}
			\caption[]%
			{{\footnotesize urbanized segregation with $\gamma = 200, \lambda = 2.$ }}    
		\label{fig:sim1c}
		\end{subfigure}
		\begin{subfigure}[t]{0.45\textwidth}   
			\centering 
	\includegraphics[width=\textwidth]{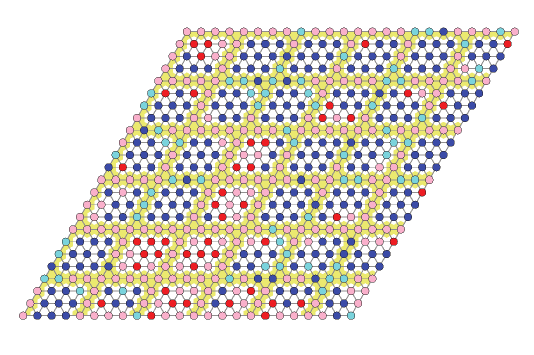}
			\caption[]%
			{{\footnotesize mitigation of segregation via distributing the urban sites, with $\gamma = 200, \lambda = 2$.}}  \label{fig:sim1d}
		\end{subfigure}
\caption[]{Simulations of $\mathcal{M}$ after five million iterations with 40$\%$ poor red(pink), 10$\%$ rich red (red), 40$\%$ rich blue (blue), and 10$\%$ poor blue (cyan) and a $46\%$ fraction of the urban sites.}\label{fig:sim3}
\end{figure*}

\begin{theorem}[\textbf{Centralized Urbanization of Poverty}] \label{urb1}
	If $\gamma > 16^{\frac{3(3\epsilon + 2)}{2\epsilon^2}}$ and $\lambda > 1$, with the centralized urban sites, when $n$ is sufficiently large, then for $\mathcal{M}$, configurations drawn from distribution $\pi$ have $\epsilon-$urbanization of poverty with probability at least $1 - \xi_{1}^{n}$, where $0< \xi_{1} < 1$.
\end{theorem}
\noindent	To prove this theorem, we first define $\Omega_{\neg \text{urb}}$ to be the set of the configurations that do not have $\epsilon-$urbanization of poverty. It suffices to show $\pi( \Omega_{\neg \text{urb}}) \leq \xi_{1}^{n}$.
	We use Peierls argument using a mapping from non-urbanized configurations to urbanized configurations, along with appropriate \textit{bridging}, to show that the image of the map has an exponentially higher probability than their preimages. With careful counting, this lets us conclude that non-urbanized configurations are exponentially less likely than  urbanized ones, even though there are many more non-urbanized configurations and some of those configurations can have large probability weights in terms of $\lambda^{-h(\sigma)}$.

	While similar to \cite{miracle2011clustering, cannon2019local}, the addition of a wealth in the model requires significantly modifying the bridge systems to encode both race and wealth for each agent using a $\delta-$race\_wealth bridge system, as specified in Section 1.2 of SI.
	{Here we have to carefully account for  additional effects contributing the energy term because some configurations with large probability  $\gamma^{p(\sigma) - p(\nu)}$ can be mapped  to configurations with small weight $\lambda^{-(h(\sigma) - h(\nu))}$. }
	 This necessitates designing more careful mapping rules to balance $\frac{\pi(\sigma)}{\pi(\nu)}$ and $|f^{-1}(\nu)|$; see Section 1.3-1.5 of SI for details of the mapping.

	 \begin{proof}[Proof of Theorem \ref{urb1}]
	For any $\sigma \in \Omega_{\neg \text{urb}}$, we first construct a $\delta-$race\_wealth bridge system (see Section 1.2 of SI for definition and Figure 2a of SI for illustration) and define the mapping $f(\sigma) = (f_{5} {\circ} f_{4} {\circ} f_{3} {\circ} f_{2} {\circ} f_{1}) (\sigma)$, where $\psi = f_{1}(\sigma)$ is the richness inversion mapping (defined in Section 1.3 of SI and see Figure 2b in SI for illustration), and $\tau = f_2(\psi)$ is the color inversion mapping (defined in Section 1.3 of SI and see Figure 2b in SI for illustration). $\tau = (f_{2} {\circ} f_{1})(\sigma)$ eliminates the bridged racially heterogeneous edges and the bridged poor agents.
	For $(f_{5} {\circ} f_{4} {\circ} f_{3}) (\tau)$ (defined in Section 1.4 of SI, we first assume the urban sites are centralized, under which we recover the same ratios of each color and richness as in $\sigma$ in the centralized way defined in Section 1.5 of SI (also see Figure 3 in SI for illustrations). 
	Then the upper bounds of $h(\nu) - h(\sigma) \leq 3\alpha \sqrt{n}  - z_{\rm c}$ and $p(\sigma) - p(\nu) \leq - \delta n$ and $|f^{-1}(\nu)| \leq   (z_{\rm c}+1) 9^{\alpha \sqrt{n}}4^{(\frac{3\delta + 1}{4\delta})(z_{\rm c}+3n)} $ can be obtained from Claim 15 and 16 in Section 2 of SI. The color contour length $z_{\rm c}$ is defined in the bridge system (Section 1.2 of SI), which is the sum of the length of the contours separating the red (or pink) from the blue (or cyan) in $\sigma$ with a no more than $\delta-$fraction omission. Finally, substituting \eqref{pi} and our other bounds into the Peierls argument yields
	\begin{align*}
		\pi(\Omega_{\neg \text{urb}}) \nonumber &= \sum_{\sigma \in \Omega_{\neg \text{urb}}} \pi(\sigma) \leq \sum_{\nu \in \Omega} \sum_{\sigma \in f^{-1}(\nu)} \pi(\sigma)\\ &= \sum_{\nu \in \Omega} \pi(\nu) \frac{\sum_{\sigma \in f^{-1}(\nu)} \pi(\sigma)}{\pi(\nu)}\nonumber\\&\leq \sum_{\nu \in \Omega} \pi(\nu) \sum_{\sigma \in f^{-1}(\nu)} \lambda^{h(\nu) - h(\sigma)} \gamma^{p(\sigma) - p(\nu)} \nonumber \\
		&\leq  \sum_{\nu \in \Omega} \pi(\nu) \sum_{z_{\rm c} = \sqrt{r\cdot n}}^{3n} a(n) (z_{\rm c}+1) (\frac{4^{\frac{3\delta + 1}{4\delta}}}{\lambda})^{z_{\rm c}} (\frac{64^{\frac{3\delta + 1}{4\delta}}}{\gamma^{\delta}})^{n} ,
	\end{align*}
	where $a(n) \triangleq (3n + 1) \cdot 9^{\alpha\sqrt{n}}n^{6}$, $z_{\rm c} \leq 3n$ is because the color contour length is upper bounded by the sum of all edges of $G_{\triangle}$, and $z_{\rm c} \geq \sqrt{r \cdot n}$ is due to the triangular lattice geometry, which is proved in Lemma 2.1 in \cite{cannon2016markov}.
	
	If $\lambda \geq 4^{\frac{3\delta + 1}{4\delta}}$, as long as $\gamma^{\delta/3} >  4^{\frac{3\delta + 1}{4\delta}}$, the sum will be exponentially small for sufficiently large $n$. Or if $ 1 \leq \lambda < 4^{\frac{3\delta + 1}{4\delta}}$, the sum further yields $\pi(\Omega_{\neg \text{urb}})  \leq n^{6} \cdot (3n +1) \cdot 9^{\alpha\sqrt{n}} \cdot 3n \cdot (\frac{16^{\frac{3\delta + 1}{4\delta}}}{\lambda\gamma^{\delta/3}})^{3n}.$ As long as $\gamma^{\delta/3} > 16^{\frac{3\delta + 1}{4\delta}} \geq 16^{\frac{3\delta + 1}{4\delta}}/\lambda$, the sum will still be exponentially small for sufficiently large $n$. Combining the two cases, we can see that as long as $\gamma^{\delta/3} > 16^{\frac{3\delta + 1}{4\delta}}$ and $\lambda > 1$, $\pi(\Omega_{\neg \text{urb}}) \leq \xi_{1}^{n},$ for $\xi_{1} \in (0,1).$ Substituting $\delta = \epsilon /2$ into $\gamma^{\delta/3} > 16^{\frac{3\delta + 1}{4\delta}}$ yields Theorem~\ref{urb1}.
\end{proof}

In the above theorem, to realize the urbanization of poverty under the centralized urban sites placement,  it suffices to have $\lambda \cdot \gamma^{\delta/3} > 16^{\frac{3\delta + 1}{4\delta}}$, where the wealth bias and the racial bias both contribute to the urbanization of poverty.
In contrast, when urban sites are distributed, we find competing effects  between $\gamma$ and $\lambda$, whereby urbanization of poverty is achieved when $\gamma$ is strictly larger than $\lambda$.  See detailed proofs in Section 3 of SI.
\begin{corollary}[\textbf{Distributed Urbanization of Poverty}] \label{urb2}
	If $\gamma > 4^{\frac{3(3\epsilon + 2)}{2\epsilon^2}} \cdot \max\{\lambda^{6/\epsilon}, 4^{\frac{3(3\epsilon + 2)}{2\epsilon^2}} \}$ and $\lambda > 1$, with the distributed urban sites, when $n$ is sufficiently large, then for $\mathcal{M}$, configurations drawn from distribution $\pi$ have $\epsilon-$urbanization of poverty with probability at least $1 - \xi_{1}^{n}$, where $0< \xi_{1} < 1$. 
\end{corollary}

On the other hand, when the incentives for the poor agents to occupy urban sites are small, urbanization of poverty will not occur.  In particular, 
we prove in Theorem~\ref{deurb} that for any $\lambda > 1$, if $\gamma > 1$ but is smaller than a threshold, it is exponentially unlikely we will observe urbanization of poverty at stationary under certain demographic parameter choices. See the proof details in Section 4 of SI.
\begin{theorem}[\textbf{Dispersion of Poverty}]\label{deurb}
	Given centralized urban sites, $p < c < p + \epsilon$, $r_{\rm p} < r_{\rm r} - \delta$ and $b_{\rm p} < b_{\rm r} - \delta$,  for any $\lambda > 1$, if $ 1< \gamma < (\frac{r-\delta}{r_{\rm p}})^{\frac{r_{\rm p}}{p}} (\frac{b-\delta}{b_{\rm p}})^{\frac{b_{\rm p}}{p}} / 2,$ when $n$ is sufficiently large, then for $\mathcal{M}$, configurations drawn from distribution $\pi$ have $\epsilon-$urbanization of poverty with probability at most $\xi_{2}^{n}$ for some constant $0< \xi_{2}< 1$ and $\delta = \frac{\epsilon}{2}.$ 
\end{theorem}

\section{Urbanized Racial Segregation}
\label{summary1}
Next, we explore conditions that lead to {\it urbanized racial segregation,} where the large regions in the urban sites are occupied by poor agents of predominantly one race. 
We define \textbf{$(\beta, \delta)-$segregation} as follows.
\begin{definition}\label{segjust}
	For $\beta > 4\sqrt{r}$ and $\delta \in (0, \frac{1}{2})$, a city configuration $\sigma$ is said to be \textbf{$(\beta, \delta)-$segregated} if there exists a subset of agents $R$ such that:
	\begin{itemize}
		\item there are at most $\beta\sqrt{n}$ racially heterogeneous edges of $\sigma$ with exactly one endpoint in $R$;
		\item the number of red agents in $R$ is at least $r n - \delta n$.
	\end{itemize}	
\end{definition}
\noindent The parameter $\delta$ is the tolerance for having agents of the other color within the red region $R$, with smaller $\delta$ corresponding to a increased segregation.
If one color class has fewer than $\delta n$ agents, then the entire configuration space will be $(\beta, \delta)-$segregated, with $R = \varnothing,$ or $\bar{R} = \varnothing.$  We require that each color class has more than $\delta n$ agents and, accordingly, we need $\delta < 1/2.$
The parameter $\beta$ controls how small the boundary is between the red region $R$ and the rest of the configuration, and the minimal value $4\sqrt{r}$ corresponds to the extremal case where the red region forms a homogeneous hexagonal cluster.   We say that a configuration $\sigma$ is {\bf integrated} if the city is not $(\beta, \delta)$-segregated for any $\beta$ and $\delta$.

\begin{definition} \label{urbanized red} 
	For $\beta > 4\sqrt{r},$ and $\epsilon \in (0, \frac{1}{2})$, we say a city has \textbf{$(\beta, \epsilon)-$urbanized segregation} if it is both \textit{$(\beta, \epsilon)-$segregated} and has \textit{$\epsilon-$urbanization of poverty}.
\end{definition}
Like Definition~\ref{segjust}, $\beta$ here controls how small the boundary is between the red region $R$ and the rest of the city, while $\epsilon$ expresses the tolerance for having agents of the wrong color within the monochromatic color regions or having rich agents on the urban sites.
In the following theorem, we show that for large enough $\lambda$ and $\gamma$,  with high probability, $\mathcal{M}$ leads to urbanized segregation. See Figure~\ref{fig:sim1c} for simulated visualizations.

\begin{theorem} [\textbf{Urbanized Racial Segregation}]\label{urb red thm}
	With centralized urban sites, $\lambda > 3^{\frac{\alpha}{\beta}} 4^{\frac{3\delta+1}{4\delta}}$, $\gamma^{\delta/3} > 4^{\frac{3\delta+1}{4\delta}}$, and $n$ sufficiently large, configurations from $\mathcal{M}$ drawn from distribution $\pi$ have $(\beta, \epsilon)-$urbanized  segregation  with probability at least $1 - \xi^{\sqrt{n}}$ for some constant $0< \xi < 1,$ and $\delta = \frac{\epsilon}{2}.$
\end{theorem}

\begin{proof}[Proof of Theorem \ref{urb red thm}]
	First, we define $U_{\beta, \epsilon} \subset \Omega$ to be the configurations that have $(\beta, \epsilon)-$urbanized segregation. To prove Theorem \ref{urb red thm}, it suffices to prove $\pi(\Omega \setminus U_{\beta, \epsilon}) \leq \xi^{\sqrt{n}},$ where $\xi \in (0, 1).$ 
	We can further divide $\Omega \setminus U_{\beta, \epsilon}$ into two parts: $\Omega_{\neg \text{urb}}$ that do not have $\epsilon-$urbanization of poverty, and $\Omega_{\text{urb} \wedge \neg \text{seg}}$ that have $\epsilon-$urbanization  of poverty and do not have $(\beta, \epsilon)-$segregation. Thus it suffices to prove $\pi( \Omega_{\neg \text{urb}}) \leq \xi_{1}^{n}$ and $\pi(\Omega_{\text{urb} \wedge \neg \text{seg}}) \leq \xi_{0}^{\sqrt{n}}$, for $0 < \xi_{1}, \xi_{0} < 1$. 
	It follows from the proof of Theorem \ref{urb1} that if $\lambda \geq 4^{\frac{3\delta+1}{4\delta}}$, as long as $\gamma^{\delta/3} > 4^{\frac{3\delta+1}{4\delta}},$ $\pi( \Omega_{\neg \text{urb}}) \leq \xi_{1}^{n}$.
	It is proved in Claim 20 in Section 5 of SI that If $\lambda > 3^{\frac{\alpha}{\beta}} 4^{\frac{3\delta+1}{4\delta}}$, $\pi(\Omega_{\text{urb} \wedge \neg \text{seg}}) \leq \xi_{0}^{\sqrt{n}}$ for some $\xi_0 \in (0, 1).$
	Combining the two parts, to have $\pi(\Omega \setminus S_{\beta, \epsilon})$ to be exponentially small for large $n$, it suffices to have $\lambda > 3^{\frac{\alpha}{\beta}} 4^{\frac{3\delta+1}{4\delta}}$ and $\gamma^{\delta/3} > 4^{\frac{3\delta+1}{4\delta}}$.
\end{proof}

To complement Theorem~\ref{urb red thm},  we prove that for large enough $\gamma$ but $\lambda>1$ below a threshold, we will  likely observe urbanization of poverty but {\it racial integration} outside the urban area under certain demographic parameter choices. The proof technique is very similar to the proof of Theorem~\ref{deurb}. See Section 6 of SI for proof details.  A special case of Theorem \ref{mix2} is shown in Remark~\ref{seg_not_hom}, where segregation of poor red agents occurs inside the urban area and racial integration occurs outside. See Figure~\ref{fig:sim1d} for simulated visualizations.
\begin{theorem}[\textbf{Coexistence of Urbanization and Racial Integration}] \label{mix2}
	With the centralized urban sites, for the demographics choices such that $(\frac{p-\delta}{r_{\rm p}})^{r_{\rm p}} (\frac{1-p-\delta}{r_{\rm r}})^{r_{\rm r}} > 2^r$,  if $ 1< \lambda^3 < (\frac{p-\delta}{b_{\rm p}})^{b_{\rm p}} (\frac{1-p-\delta}{b_{\rm r}})^{b_{\rm r}} / 2^r,$ and $\gamma^{\delta/3} > 16^{\frac{3\delta + 1}{4\delta}},$ when $n$ is sufficiently large, then for $\mathcal{M}$, configurations drawn from distribution $\pi$ have $\epsilon-$urbanization of poverty and are integrated outside the urban area with probability at least $1 - \xi_{3}^{n}$ for some constant $0< \xi_{3}< 1$ and $\delta = \frac{\epsilon}{2}.$  
\end{theorem}

\begin{remark} \label{seg_not_hom}
	As a special case when the size of the urban sites can roughly accommodate all of poor agents, where $p < c < p + \epsilon,$ if the demographics satisfies Theorem \ref{mix2} with $b_{\rm p} \leq m \epsilon$, then under the same bias parameter choices as Theorem \ref{mix2}, then with high probability the stationary configuration will have urbanized segregation of poor red agents, where the density of poor red on the urban sites is at least $1 - (m+1)\epsilon,$ and racial integration of the rich outside the urban area.
\end{remark}

\begin{proof}[Proof of Theorem \ref{mix2}]
    	We define $\Omega_{\rm{urb} \wedge \neg \text{seg}} \subset \Omega$ to be the configurations that have $\epsilon-$urbanization of the poor and $(\beta, \delta)-$integration outside the urban area. 
	It suffices to show that with all but exponentially small probability, a sample drawn from \eqref{pi} is not in $\Omega_{\rm{urb} \wedge \neg \text{seg}} $: $\pi(\Omega \setminus \Omega_{\rm{urb} \wedge \neg \text{seg}} ) \leq \xi_3^{\sqrt{n}},$ where $\xi_3 \in (0, 1)$ and $n$ is sufficiently large.

	We can further divide the configuration space $\Omega \setminus \Omega_{\rm{urb} \wedge \neg \text{seg}}$ into two parts: the set of configurations $\Omega_{\neg \text{urb}}$ that do not have $\epsilon-$urbanization of poverty, and the set of configurations $\Omega_{\text{urb} \wedge \text{seg}}$ that have $\epsilon-$urbanization  of poverty and $(\beta, \delta)-$segregation. Since $\Omega \setminus \Omega_{\rm{urb} \wedge \neg \text{seg}} = \Omega_{\neg \text{urb}} + \Omega_{\text{urb} \wedge \text{seg}}$, to prove $\pi(\Omega \setminus \Omega_{\rm{urb} \wedge \neg \text{seg}}) \leq \xi_3^{\sqrt{n}},$ it suffices to prove $\pi( \Omega_{\neg \text{urb}}) \leq \xi_{1}^{n}$ and $\pi(\Omega_{\text{urb} \wedge \text{seg}}) \leq \xi_{0}^{\sqrt{n}}$, for some constant $0 < \xi_{1}, \xi_{0} < 1$ and sufficiently large $n$.
	
	It follows from Theorem \ref{urb1} that for $\gamma^{\delta/3} > 16^{\frac{3\delta + 1}{4\delta}}$ and $\lambda > 1$, $\pi(\Omega_{\neg \text{urb}}) \leq \xi_1^{n}$ for some $\xi_1 \in (0, 1).$ 
	To prove the second part $\pi(\Omega_{\text{urb} \wedge \text{seg}}) \leq \xi_{0}^{\sqrt{n}}$, for each $\sigma \in \Omega_{\text{urb} \wedge \text{seg}},$ we construct a $\delta-$color bridge system (see Section 1.2 of SI for definition). Then we define the mapping $s = (s_2 \circ f_2)(\cdot)$: we do the color inversion and obtain $\tau = f_2(\sigma);$ next for $\tau$, we randomly flip the cyan to pink until the right number of the pink, and we randomly flip the blue to red until the right number of the red and obtain $\nu = s_2(\tau).$
	
	Finally, we define a weighted bipartite graph $G(\Omega_{\text{urb} \wedge \text{seg}}, \Omega,E)$ with an edge of weight $\pi(\sigma)$ between $\sigma \in \Omega_{\text{urb} \wedge \text{seg}}$ and $\nu \in \Omega$. The total weight of edges is
	\begin{align}\label{LHS2}
		\sum_{\sigma \in \Omega_{\text{urb} \wedge \text{seg}}} &\pi(\sigma) \cdot |S(\sigma)| \nonumber \\ &\geq \pi(\Omega_{\text{urb} \wedge \text{seg}})  (\frac{p-\delta}{r_{\rm p}})^{(r_{\rm p}-\delta)n} (\frac{1-p-\delta}{r_{\rm r}})^{(r_{\rm r}-\delta)n}.
	\end{align}
	On the other hand, the weight of the edges is at most 
	\begin{align} \label{RHS2}
		&\sum_{\nu \in \Omega} \sum_{\sigma \in s^{-1}(\nu)}\max_{\sigma \in \Omega_{\text{urb} \wedge \text{seg}}}{\pi(\sigma)} \nonumber \\ &=  \sum_{\nu \in \Omega} \pi(\nu) \sum_{\sigma \in s^{-1}(\nu)} \frac{\max_{\sigma \in \Omega_{\text{urb} \wedge \text{seg}}}{\pi(\sigma)}}{\pi(\nu)}  |s^{-1}(\nu)|\nonumber\\
		&\leq \sum_{\nu \in \Omega} \pi(\nu) \sum_{z_{\rm c = \sqrt{rn}}}^{\beta \sqrt{n}} l \cdot \lambda^{\max(h(\nu)-h(\sigma))} \gamma^{\max(p(\sigma)-p(\nu))}  4^{\frac{3\delta + 1}{4\delta}z_{\rm c}} 2^{rn}
		\nonumber \\&\leq \sum_{z_{\rm c} = \beta_{\min}\sqrt{n}}^{\beta \sqrt{n}} l \cdot (\frac{4^{\frac{3\delta + 1}{4\delta}}}{\lambda})^{z_{\rm c}} \lambda^{3n}2^{rn}.
	\end{align}
	where $l \triangleq (z_{\rm c}+1)3^{\alpha \sqrt{n}}$, and the inequalities in Claims 23-25 from Section 6 of SI have been substituted in the above derivation. Combining equations~\eqref{LHS2} and \eqref{RHS2},  we find
	\begin{align}\label{inequal_coo}
		&\pi(\Omega_{\text{urb} \wedge \text{seg}})  (\frac{p-\delta}{r_{\rm p}})^{(r_{\rm p}-\delta)n} (\frac{1-p-\delta}{r_{\rm r}})^{(r_{\rm r}-\delta)n} 
		\nonumber \\ &\leq \sum_{z_{\rm c} = \beta_{\min}\sqrt{n}}^{\beta \sqrt{n}} (z_{\rm c}+1)3^{\alpha \sqrt{n}} (\frac{4^{\frac{3\delta + 1}{4\delta}}}{\lambda})^{z_{\rm c}} \lambda^{3n}2^{rn}.
	\end{align}
	For large enough $n$, to have $\pi(\Omega_{\text{urb} \wedge \text{seg}})  \leq \xi_3^{n}$ for some $\xi_3 \in (0, 1)$, it suffices to have \begin{align*}
\lambda^{3n}2^{rn} &< (\frac{p-\delta}{r_{\rm p}})^{(r_{\rm p}-\delta)n} (\frac{1-p-\delta}{r_{\rm r}})^{(r_{\rm r}-\delta)n} \\ &< (\frac{p-\delta}{r_{\rm p}})^{r_{\rm p}n} (\frac{1-p-\delta}{r_{\rm r}})^{r_{\rm r}n},
\end{align*}
which can be rewritten as 
	$$\lambda^3 < (\frac{p-\delta}{r_{\rm p}})^{r_{\rm p}} (\frac{1-p-\delta}{r_{\rm r}})^{r_{\rm r}} / 2^r.$$
	Since $\lambda > 1,$ to make the right hand side of the above inequality greater than one, it suffices to have $(\frac{p-\delta}{r_{\rm p}})^{r_{\rm p}} (\frac{1-p-\delta}{r_{\rm r}})^{r_{\rm r}} > 2^r.$
	Combining the above parameter choices with Theorem \ref{urb1} requires $\gamma^{\delta/3} > 16^{\frac{3\delta + 1}{4\delta}}$ and $\lambda >1$,  proving Theorem~\ref{mix2}.
\end{proof}

\section{Integration in Cities with Distributed Urban Sites}
{As we've shown, when urban sites are centralized,  the wealth and homophily biases align to cause segregation, as shown in Theorem~\ref{urb red thm}. However, when the placement of urban sites is  {\it distributed}, the racial and wealth biases will work against each other, and we will get integration if the influence of the wealth bias exceeds the homophily bias. See Figure~\ref{fig:sim1d} for simulations. 
\begin{theorem} \label{miti}
	In a city with $|\mathcal{U}| = cn$ urban sites evenly partitioning in the city (like we find with bus routes), a small number of poor blue agents, with $b_{\rm p} < \hat{b}_{\rm p}$, and any $\lambda>1$, if $\gamma^{\hat{b}_{\rm p} - b_{\rm p}} > \lambda^{2c}64^{\frac{3\delta + 1}{4\delta}}$,  and $n$ sufficiently large, then configurations drawn from distribution $\pi$ will be $(\beta, \delta)-$segregated with exponentially small probability $\xi_4^{n}$, for some constant $0< \xi_4< 1.$
\end{theorem}
	 Hence, integration occurs because no matter how large the homophily bias weight $\lambda^{-(h(\sigma) - h(\nu))}$ is, as long as the  energy term arising from the wealth bias   $\gamma^{p(\sigma) -p(\nu)}$ is larger, then the stationary distribution will be very unlikely to be segregated.} See below for the proof.

\begin{proof}[Proof of Theorem \ref{miti}]
	We define the configuration space $S_{\beta, \delta}$ to be the set of configurations that are $(\beta, \delta)-$segregated. To prove Theorem \ref{miti}, it suffices to prove $\pi(S_{\beta, \delta}) \leq \xi_4^{n}$, where $\xi_4 \in (0,1).$
	The bridging and the mapping $\nu = f(\sigma) = (f_5 {\circ} f_4 {\circ} f_3 {\circ} f_3 {\circ} f_1)(\sigma)$ are defined as the following: we first construct a $\delta-$race\_wealth bridge system for $\sigma \in S_{\beta, \delta}$ (see Section 1.2 of SI for details). Then we do richness inversion and color inversion like defined in $f_1(\cdot)$ and $f_2(\cdot)$. 
    Then we do the color and richness recovery $\nu = (f_5{\circ}f_4{\circ}f_3{\circ})(\tau)$ in the distributed way as specified in Section 1.5 of SI. After the mapping, it is proved in Claim 26 of SI that for any $\sigma \in S_{\beta, \delta}$ with a given bridged color contour length $z_{\rm c}$, $h(\nu) - h(\sigma) \leq 2cn  - z_{\rm c}$ and $p(\sigma) - p(\nu) \leq b_{\rm p}n - \hat{b}_{\rm p}n$, where $\hat{b}_{\rm p} \triangleq \min\{c, p\}-(r+\delta)c - 2\delta.$ See proof details of Claim 26 from Section 7 of SI.

	For a given color contour length $z_{\rm c},$ for any $\nu = f(\sigma),$ the number of preimages follows from Claim 16. Similarly, we use Peierls argument \eqref{map1}, substituting the related bounds into which yields
	\begin{align*}
		&\pi(S_{\beta, \delta}) \leq \nonumber \\ &\sum_{\nu \in \Omega} \pi(\nu) \sum_{z_{\rm c} = \sqrt{r\cdot n}}^{\beta\sqrt{n}} n^{6}(3n + 1)(z_{\rm c}+1) 9^{\alpha\sqrt{n}}(\frac{4^{\frac{3\delta + 1}{4\delta}}}{\lambda})^{z_{\rm c}} (\frac{\lambda^{2c}64^{\frac{3\delta + 1}{4\delta}}}{\gamma^{\hat{b}_{\rm p} - b_{\rm p}}})^{n},
	\end{align*}
	where  $z_{\rm c} \geq \sqrt{r \cdot n}$ is due to the triangular lattice geometry, which is proved in Lemma 2.1 in \cite{cannon2016markov}, and $z_{\rm c} \leq \beta\sqrt{n}$ is due to $\sigma \in S_{\beta, \delta}$ and the definition of $(\beta, \delta)-$segregation.
	If $b_{\rm p} <  \hat{b}_{\rm p}$ and $\gamma^{\hat{b}_{\rm p} - b_{\rm p}} > \lambda^{2c}64^{\frac{3\delta + 1}{4\delta}}$, the sum will be exponentially small given large enough $n$, which means $\pi(S_{\beta, \delta}) \leq \xi_4^{n}$ for some $\xi_4 \in (0, 1).$
\end{proof}

\begin{remark}\label{inequality}
	If the number of poor blue agents satisfies $b_{\rm p} < \hat{b}_{\rm p} \triangleq \min\{c, p\}-(r+\delta)c - 2\delta,$ we can conclude the ratio between poor blue and poor red agents is smaller than the ratio between the blue and the red: $\frac{b_{\rm p}}{r_{\rm p}}  < \frac{b}{r},$ which is understood as income inequality.
\end{remark}
\begin{proof}
	If $c \leq p$, then it follows that $b_{\rm p} < c -(r+\delta)c - 2\delta = (b-\delta)c - 2\delta < (b-\delta)c < (b-\delta)p < b\cdot p$, which can be written as $\frac{b_{\rm p}}{r_{\rm p}} < \frac{b \cdot p}{r_{\rm p}}= b + b \cdot \frac{b_{\rm p}}{r_{\rm p}}$. Hence we can get $\frac{b_{\rm p}}{r_{\rm p}}  < \frac{b}{r}.$ If $p \leq c$, it follows that $b_{\rm p} < p -(r+\delta)c - 2\delta < p -(r+\delta)p - 2\delta < p (b - \delta) < b \cdot p$. Hence the same conclusion $\frac{b_{\rm p}}{r_{\rm p}}  < \frac{b}{r}$  follows. 
\end{proof}

\begin{remark}
    Although Theorem \ref{miti} is proved under the lattice-shaped urban sites shown in Figure~\ref{1b}. But urban sites can also be distributed in other ways and a similar proof will follow, like many (order of $\mathcal{O}(n)$) small clusters of disconnected urban sites.
\end{remark}

\section{Simulations}
\begin{figure*}
    	\begin{subfigure}[t]{0.45\textwidth}  
		\centering 
		\includegraphics[width=\textwidth]{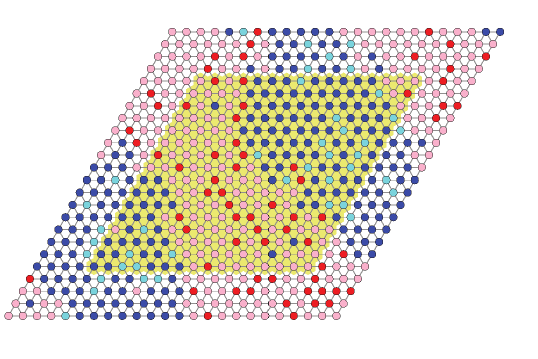}
		\caption[]%
		{{\footnotesize segregation no wealth bias when $\gamma = 1, \lambda = 2$.}}
		\label{fig:sim1e}
    	\end{subfigure}
        \hfill
    	\begin{subfigure}[t]{0.45\textwidth}   
		\centering 
	\includegraphics[width=\textwidth]{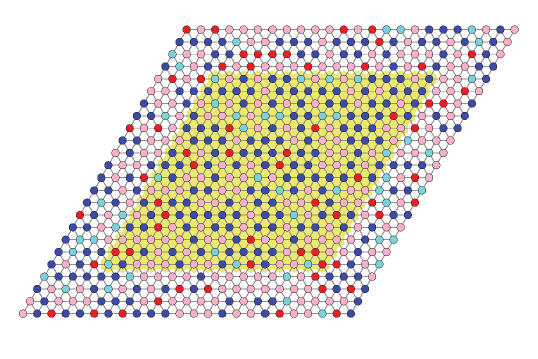}
		\caption[]%
		{\footnotesize integration with $\gamma = 1.01, \lambda = 1.01$.}
		\label{fig:sim1f}
    	\end{subfigure}
\caption[]{Simulations of $\mathcal{M}$ after five million iterations with 40$\%$ poor red(pink), 10$\%$ rich red (red), 40$\%$ rich blue (blue), and 10$\%$ poor blue (cyan) and a $46\%$ fraction of the urban sites.}\label{fig:sim4}
\end{figure*}

\begin{figure}[b]
		\centering
	\includegraphics[width=0.4\textwidth]{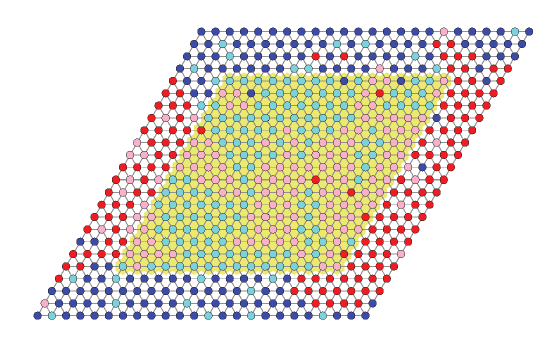}  
		\label{fig:sim2a}
	\caption[]{{Urbanized segregation without income inequality under strong wealth and racial biases ($\gamma = 200, \lambda = 1.01$). The racial and wealth distribution is 25$\%$ poor red, 25$\%$ rich red, 25$\%$ rich blue, and 25$\%$ poor blue. }
	}\label{fig:sim2}
\end{figure}

We supplement the  theorems with simulations of $\mathcal{M}$, shown in Figure~\ref{fig:sim1}, \ref{fig:sim3}, and \ref{fig:sim4}, for a city with income inequality starting from random initial locations of agents. Figure~\ref{fig:sim1} compares configurations after running $\mathcal{M}$  for the same number of iterations, varying only the values of $\lambda$, $\gamma$, and the placement of urban sites. Note that the parameter settings for $\lambda$ and $\gamma$ in the simulations are better than in our theorems, confirming that our bounds are likely not tight.

Figure~\ref{fig:sim1a} demonstrates Theorems~\ref{urb1} and~\ref{mix2} showing the coexistence of urbanization of poverty and racial integration outside the urban area under strong wealth bias but slight racial bias. 
Specially, since the chosen urban area can  accommodate all the poor agents in a city and the city has severe income inequality, Figure~\ref{fig:sim1a} can also be viewed as a verification of Remark~\ref{seg_not_hom}, showing segregation of poor red agents in the urban area and integration outside.  Individuals in Figure~\ref{fig:sim1a} have small racial biases, so the wealth biases can also drive racial segregation under the centralized placement of urban sites.
Figure~\ref{fig:sim1b} verifies Corollary \ref{urb2}, showing the urbanization of poverty with distributed urban sites. Compared with Figure~\ref{fig:sim1a}, the pink cluster gets dispersed via the distributed urban sites.
Figure~\ref{fig:sim1c} verifies Theorem~\ref{urb red thm}, showing the urbanized segregation. Due to income inequality, where most of the poor agents are red, we can see that the pink predominantly occupies the urban area. In contrast, in Figure~\ref{fig:sim2},  when there is income equality across races, we can see urbanized segregation and roughly the same amount of poor red and poor blue agents occupying the urban sites.
Figure~\ref{fig:sim1d} demonstrates Theorem~\ref{miti}, showing the mitigation of segregation via distributing the urban sites in the existence of agents' strong racial bias, which should lead to Figure~\ref{fig:sim1c} urban sites are not distributed. Compared with Figure~\ref{fig:sim1b}, whose segregation level is even smaller, the difference is that agents in Figure~\ref{fig:sim1b} have little racial bias, whereas in Figure~\ref{fig:sim1d} each agent has strong racial bias.
Figures~\ref{fig:sim1e} and~\ref{fig:sim1f} provide baselines of the main work. Figure~\ref{fig:sim1e} shows segregation under strong racial bias without wealth bias, which was proved in~\cite{cannon2019local, li2021programming}. Figure~\ref{fig:sim1f} shows integration under little racial and wealth bias, which is proved in Theorem~\ref{deurb}.

\section{Conclusions}
lIn this work, we consider the interplay of race and socioeconomic status  by introducing a heterogeneous stochastic Schelling model with urban sites as incentives for the poor individuals. We show that compact and centralized urban sites, like in a city center, encourage poor agents to cluster centrally, while infrastructure that is well distributed, like a large grid of bus routes spanning the entire city, tends to disperse the low-income agents.  Understanding the effects of these two scenarios on segregation can be helpful for understanding how to best distribute public amenities to help mitigate segregation.

We find that centralized infrastructure simultaneously causes  an “urbanization of poverty” (i.e., occupation of urban sites primarily by poor agents) and segregation when both the homophily and  incentives drawing poor agents to urban sites are large enough. Moreover, if there is {\it income inequality} where one race has a significantly higher proportion of poor agents, when homophily preferences are small and incentives drawing the poor individuals to urban sites are sufficiently large,  we get  racial {\it segregation} on urban sites and {\it integration} on non-urban sites, with high probability. However, we find there is overall mitigation of segregation (on urban and non-urban sites) whenever the urban sites are spatially distributed throughout the lattice and the incentives drawing poor agents to the urban sites exceed the homophily preference. We prove that in this case, no matter how strong homophily preferences are, it will be exponentially unlikely that a configuration chosen from stationarity will have large, homogeneous clusters of similarly colored agents, thus promoting integration in the city.  These findings suggest that deliberate urban planning can mitigate or enhance segregation.

We note that there are many limitations of the heterogeneous Schelling model studied here, with many variants worth considering. For instance, in addition to amenities that are only preferred by poor agents, we can introduce other amenities  preferred by rich agents \cite{leroy1983paradise}, such as fine arts and various recreational services. We intend to introduce such incentives for the rich in subsequent work.

Naturally, the model considered here is an abstraction that oversimplifies biases and ignores many factors affecting segregation and incentives in the real world. Many other important factors, such as housing prices and individuals' preferences for higher-income neighbors, are not captured in our model. Nonetheless, we believe that this simple model can provide insight into how socioeconomic incentives might worsen or mitigate segregation through the allocation of urban amenities. 
To supplement the theoretical findings presented here, we also are exploring relevant demographic data in cities across the United States \cite{zhao2022correlation}. After collecting  national data, we find some positive correlations between cities with more distributed amenities and better racial integration, as found in our model.

\bibliographystyle{ACM-Reference-Format}
\bibliography{references}

\pagebreak
\begin{center}
\textbf{\large Supplemental Information}
\end{center}
\setcounter{equation}{0}
\setcounter{figure}{0}
\setcounter{table}{0}
\setcounter{page}{1}
\makeatletter
\renewcommand{\theequation}{S\arabic{equation}}
\renewcommand{\thefigure}{S\arabic{figure}}
\renewcommand{\bibnumfmt}[1]{[S#1]}
\renewcommand{\citenumfont}[1]{S#1}

\setcounter{section}{0}

\section{Technical Summary and Proof Details}

\subsection{Detailed Balance Proof that $\pi$ is the Stationary Distribution of $\mathcal{M}$} \label{A1}
\begin{proof}
	Consider any two configurations $\sigma$ and $\tau$ that differ by one swap transition between agent $i$ and $j$. It follows from $\mathcal{M}$ that the probability of transitioning from $\sigma$ to $\tau$ is
	\begin{align*}
		M(\sigma, \tau) = \frac{1}{\binom{n}{2}} \lambda^{- N_{\sigma}(i) - N_{\sigma}(j)} \gamma^{- u_{\sigma}(i) - u_{\sigma}(j)}.
	\end{align*}
	A similar analysis shows 
	\begin{align*}
		M(\tau, \sigma) = \frac{1}{\binom{n}{2}} \lambda^{- N_{\tau}(i) - N_{\tau}(j)} \gamma^{- u_{\tau}(i) - u_{\tau}(j)}.
	\end{align*}
	For the finite torus triangular lattice with $n$ vertices, the total number of racially heterogeneous edges $h(\sigma) = 3n - e(\sigma),$ where $e(\sigma)$ is the number of homogeneous edges. Hence, substituting $\pi(\sigma) = \frac{\lambda^{-h(\sigma)}\gamma^{p(\sigma)}}{Z}$ into $\pi(\sigma)M(\sigma, \tau)$ yields
	\begin{align*}
		\pi(\sigma)M(\sigma, \tau) &= \frac{\lambda^{-h(\sigma)- N_{\sigma}(i) - N_{\sigma}(j)} \gamma^{p(\sigma) - u_{\sigma}(i) - u_{\sigma}(j)}}{\binom{n}{2} \cdot Z} \\ &= \frac{\lambda^{-3n + e(\sigma)- N_{\sigma}(i) - N_{\sigma}(j)} \gamma^{\left (\sum_{k \in \mathcal{A}}u_{\sigma}(k) \right ) - u_{\sigma}(i) - u_{\sigma}(j)}}{\binom{n}{2} \cdot Z} \\ &= \frac{\lambda^{e(\sigma)- N_{\sigma}(i) - N_{\sigma}(j)} \gamma^{\sum_{k \in \mathcal{A} \setminus \{i, j\}}u_{\sigma}(k)}}{\binom{n}{2} \cdot Z \cdot \lambda^{3n}}.
	\end{align*}
	We denote $\sum_{k \in \mathcal{A} \setminus \{i, j\}}u_{\sigma}(k) \triangleq p(\sigma \setminus \{i, j\})$, which means the total number of the poor on the urban sites after removing agent $i$ and $j$ from the configuration $\sigma$. Similarly, we denote the total number of homogeneous edges after removing agent $i$ and $j$ from the configuration $\sigma$ as $e(\sigma \setminus \{i, j\})$. If agent $i$ and $j$ are not neighbors, $e(\sigma)- N_{\sigma}(i) - N_{\sigma}(j) = e(\sigma \setminus \{i, j\})$; and if agent $i$ and $j$ are neighbors, $e(\sigma)- N_{\sigma}(i) - N_{\sigma}(j) = e(\sigma \setminus \{i, j\}) - 1.$ Hence,
	\begin{align*}
		\pi(\sigma)M(\sigma, \tau) = \left\{\begin{matrix}
			\frac{\lambda^{e(\sigma \setminus \{i, j\})} \gamma^{p(\sigma \setminus \{i, j\})}}{\binom{n}{2} \cdot Z \cdot \lambda^{3n}},  \text{$i$ and $j$ are not neighbors}\\ 
			\frac{\lambda^{e(\sigma \setminus \{i, j\}) - 1} \gamma^{p(\sigma \setminus \{i, j\})}}{\binom{n}{2} \cdot Z \cdot \lambda^{3n}},  \text{$i$ and $j$ are neighbors}
		\end{matrix}\right..
	\end{align*}
	A similar analysis shows that
	\begin{align*}
		\pi(\tau)M(\tau, \sigma) = \left\{\begin{matrix}
			\frac{\lambda^{e(\tau \setminus \{i, j\})} \gamma^{p(\tau \setminus \{i, j\})}}{\binom{n}{2} \cdot Z \cdot \lambda^{3n}},  \text{$i$ and $j$ are not neighbors}\\ 
			\frac{\lambda^{e(\tau \setminus \{i, j\}) - 1} \gamma^{p(\tau \setminus \{i, j\})}}{\binom{n}{2} \cdot Z \cdot \lambda^{3n}},  \text{$i$ and $j$ are neighbors}
		\end{matrix}\right..
	\end{align*}
	Since configurations $\sigma$ and $\tau$ only differ by the swap move between $i$ and $j$, hence after removing $i$ and $j$ from both configurations, the two configurations $\sigma \setminus \{i, j\}$ and $\tau \setminus \{i, j\}$ are the same. Thus, we can conclude $\pi(\sigma)M(\sigma, \tau) = \pi(\tau)M(\tau, \sigma).$ Since the detailed balance is satisfied, we conclude the stationary distribution $\pi$ is given by $\pi(\sigma) = \frac{\lambda^{-h(\sigma)}\gamma^{p(\sigma)}}{Z}$.
\end{proof}

\subsection{Bridge Systems} \label{A:crbridge}
Throughout the proofs, we use red and blue to represent the races, and richness of the color to represent the wealth of each agent (rich red is red; poor red is pink; rich blue is blue; poor blue is cyan). 
We first need to extend the bridging technique to expand the encoded information dimension from color only to both color and richness over the methods in \cite{miracle2011clustering, cannon2019local}, within our context. The following shows our adapted bridging technique.

\vskip.1in
\noindent{\bf Lattice Duality.} 
The hexagonal dual to the triangular lattice $G_{\triangle}$ is obtained by creating a  vertex at the centroid of each unit triangle in $G_{\triangle}$ and connecting two of these vertices if their corresponding unit triangles have a common edge (as shown in Figure \ref{2a}, the obtained hexagonal lattice is denoted by $G_{\rm{hex}}$).
Each edge $e^{'} \in G_{\rm{hex}}$ crosses a unique edge $e \in G_{\triangle}$ and separates two adjacent agents living at the of $e$. There is a bijection between edges of $G_{\triangle}$ and edges of $G_{\rm{hex}}$, associating an edge of $G_{\triangle}$ with the unique edge of $G_{\rm{hex}},$ and vice versa.

\vskip.1in
\noindent{\bf Color Contours and Color Bridges.}  If an edge $e^{'} \in G_{\rm{hex}}$ separates two agents heterogeneous in race, we call it a {\it color edge}. We define a \textit{color contour} to be made up of color edges and is a self-avoiding polygon in $G_{\rm hex}$ that never visits the same vertex twice except to start and end at the same place. The color contour is denoted in green as shown in Figure \ref{2b}.  The \textit{color bridges} are shown in dashed green. They are self-avoiding walks on $G_{\rm{hex}}$ that connect color contours to the boundary.

\vskip.1in
\noindent{\bf Richness Contours and Richness Bridges.}  If an edge $e^{'} \in G_{\rm{hex}}$ separates two agents heterogeneous in wealth, we call it a {\it richness edge}. We define a \textit{richness contour} to be made up of richness edges and is a self-avoiding polygon, which is denoted in orange as shown in Figure \ref{2c}. The \textit{richness bridges} are shown in dashed orange, which connects richness contours to the boundary.

\begin{figure} 
	\begin{subfigure}[t]{0.3\textwidth} 
		\centering
		\includegraphics [height=1.3in]{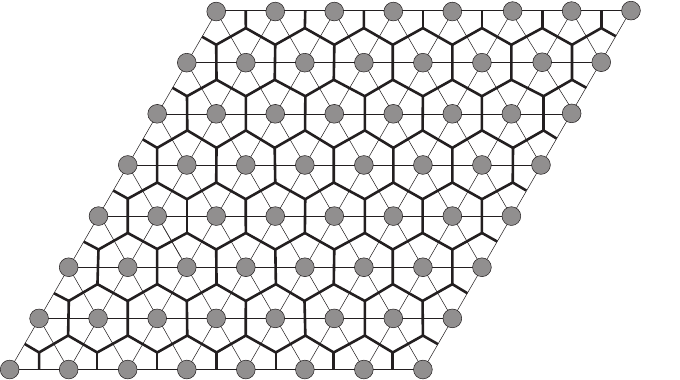}
		\captionsetup{justification=centering}
		\subcaption{}\label{2a}
	\end{subfigure}
	\begin{subfigure}[t]{0.3\textwidth} 
		\centering
		\includegraphics[height=1.3in]{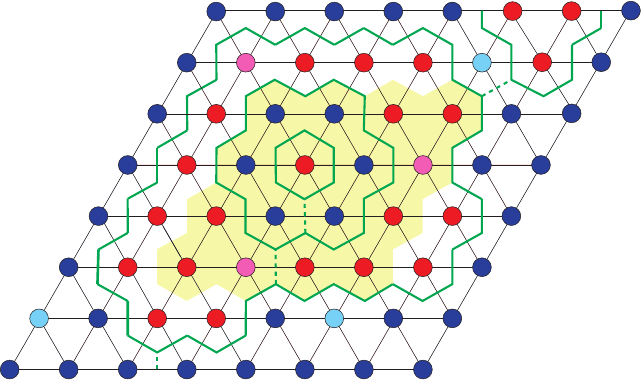}
		\captionsetup{justification=centering}
		\subcaption{}\label{2b}
	\end{subfigure}
	\begin{subfigure}[t]{0.3\textwidth} 
		\centering
		\includegraphics[height=1.3in]{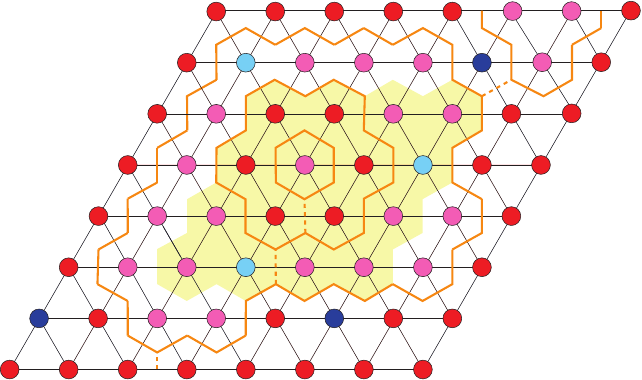}
		\captionsetup{justification=centering}
		\subcaption{}\label{2c}
	\end{subfigure}
	\caption[]{(a) The hexagon lattice $G_{\triangle}$ dual to the city lattice $G_{\triangle}$. (b) Demonstration of the color contours and color bridges. (c) Demonstration of the richness contours and richness bridges. }\label{2}
\end{figure}

\vskip.1in
\noindent{\bf $\delta-$Race\underbar\ Wealth Bridge System.}  
We first define the $\delta-$color bridge system similarly as \cite{miracle2011clustering,cannon2019local} as the following. For any color component $F$ \footnote{A color component is a maximal simply connected subset $F$ of agents where all agents in $F$ adjacent to a location not in $F$ have the same race, which we call the color of F.},  let $C_{\rm c}$ be a collection of color contours of $F$. The color bridges collection $B_{\rm c}$ connects each color contour in $C_{\rm c}$ to the boundary of $F$.  

An agent $P$ is bridged in terms of color in $F$ if there exists a path through agents of the same race as $P$ to the boundary of $F$ or a bridged color contour in $C_{\rm c}$. An agent is unbridged in terms of color if such a path does not exist. Then we define that $(B_{\rm c}, C_{\rm c})$ is a \textit{$\delta-$color bridge system} for $F$ if 
\begin{itemize}
	\item $|B_{\rm c}| \leq |C_{\rm c}| (1  - \delta)/2\delta$, where $|B_{\rm c}|$ is the total number of edges in $B_{\rm c}$ and $|C_{\rm c}|$ is the total number of edges in $C_{\rm c}$;
	\item the number of unbridged agents in terms of color in $F$ is $\leq \delta|F|$, where $|F|$ is the number of agents in $F$.
\end{itemize}
Note $\delta \in [0, 1]$ controls how much color information is omitted by the $\delta-$color bridge system, and see proof for Lemma 7.2 in \cite{cannon2019local} for the construction way of a $\delta-$color bridge system for any $F$.
\begin{lemma}[\cite{cannon2019local}]\label{I:cbr}
	For any color component $F$, there exists a $\delta-$color bridge system for $F$.
\end{lemma}

Similarly, we can also define the $\delta-$richness bridge system, and the richness component. The following lemma holds similarly.
\begin{lemma}[\cite{cannon2019local}]\label{I:rbr}
	For any richness component $F$, there exists a $\delta-$richness bridge system for $F$.
\end{lemma}

We call the joint  $\delta-$color and $\delta-$richness bridge system  a \textit{$\delta-$-race\underbar\ wealth bridge system.}
For any configuration $\sigma  \in  \Omega_{\rm t},$ we can construct a $\delta-$race\underbar\ wealth bridge system. See Figure \ref{3a} for illustrations. where at most $\delta n$ agents are not bridged in terms of color, and at most $\delta n$ agents are not bridged in terms of richness. Combining Lemma \ref{I:cbr} and Lemma \ref{I:rbr}, the following lemma holds.
\begin{lemma}\label{I:rcbr}
	For any finite region $F$, there exists a $\delta-$race\underbar\ wealth bridge system for $F$.
\end{lemma}

\vskip.1in
\noindent{\bf Crossing Contours and Non-Crossing Contours.}  The contour that touches the boundary of the defined domain is called a \textit{crossing contour}. The contour that does not touch the boundary is called a \textit{non-crossing contour}. For example, Figure \ref{2b} has one crossing color contour and three non-crossing color contours.

The sum of the number of edges of a contour is called the length of the contour. We denote the length of the bridged non-crossing contours as $y$, and length of the crossing contours as $x$. The length of the bridged non-crossing color contours is denoted by $y_{\rm c}$. The length of the crossing color contours is $x_{\rm c}$. We can also define non-crossing richness contours $y_{\rm r}$ and crossing richness contours  $x_{\rm r}$ in similar way. It follows that $y = y_{\rm c} + y_{\rm r}$ and $x = x_{\rm c} + x_{\rm r}$. We call $z_{\rm c } \triangleq x_{\rm c} + y_{\rm c}$ the bridged color contour length, and $z_{\rm r } \triangleq x_{\rm r} + y_{\rm r}$ the bridged richness contour length.


For a given bridged color or richness contour length, the following lemmas bound the number of possible bridge systems, which can be counted in a depth-first way (see proof details in \cite{cannon2019local}).
\begin{lemma}[Lemma 7.6 in \cite{cannon2019local}] \label{I:rbridge}
	For a given $z_{\rm r} = y_{\rm r} + x_{\rm r}$, there are at most $(z_{\rm r} + 1) 3^{\alpha \sqrt{n}}4^{\frac{3\delta +1}{4\delta}z_{\rm r}}$ ways of constructing a $\delta-$richness bridge system.
\end{lemma}

\begin{lemma}[Lemma 7.6 in \cite{cannon2019local}] \label{I:cbridge}
	For a given $z_{\rm c} = y_{\rm c} + x_{\rm c}$, there are at most $(z_{\rm c} + 1) 3^{\alpha \sqrt{n}}4^{\frac{3\delta +1}{4\delta}z_{\rm c}}$ ways of constructing a $\delta-$color bridge system.
\end{lemma}
Combining Lemma \ref{I:rbridge} and \ref{I:cbridge} yields the following lemma.
\begin{lemma}\label{I:bridge_num}\label{I:rcbridge}
	For a given $z_{\rm c} = y_{\rm c} + x_{\rm c}$, there are at most $(\max\{z_{\rm r} \} + 1)(z_{\rm c} + 1) 9^{\alpha \sqrt{n}}4^{(\frac{3\delta +1}{4\delta})(z_{\rm c} + \max\{z_{\rm r} \})}$ ways of constructing a $\delta-$race\underbar\ wealth bridge system.
\end{lemma}

\begin{figure} 
	\begin{subfigure}[t]{0.3\textwidth} 
		\centering
		\includegraphics[height=1.3in]{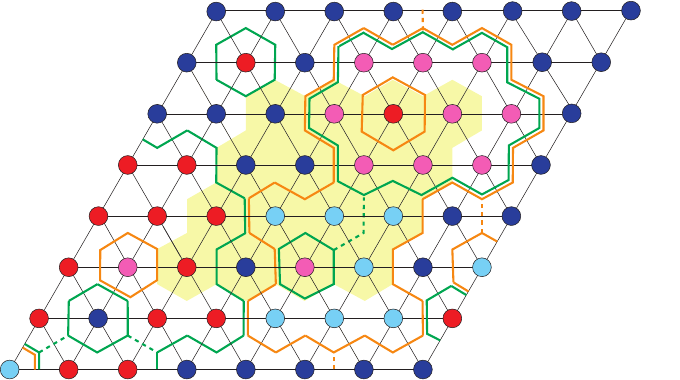}
		\captionsetup{justification=centering}
		\subcaption{}\label{3a}
	\end{subfigure}
	\begin{subfigure}[t]{0.3\textwidth} 
		\centering
		\includegraphics[height=1.3in]{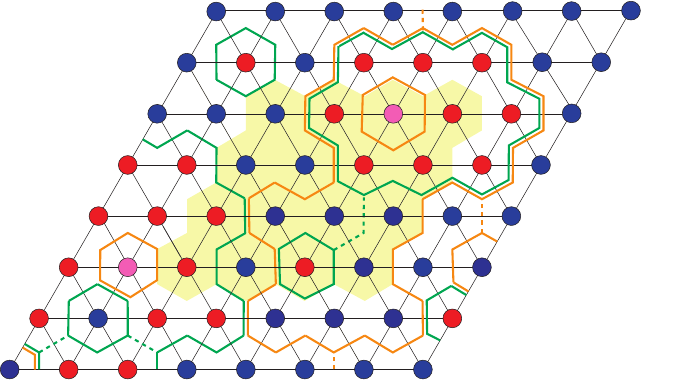}
		\captionsetup{justification=centering}
		\subcaption{}\label{3b}
	\end{subfigure}
	\begin{subfigure}[t]{0.3\textwidth} 
		\centering
		\includegraphics[height=1.3in]{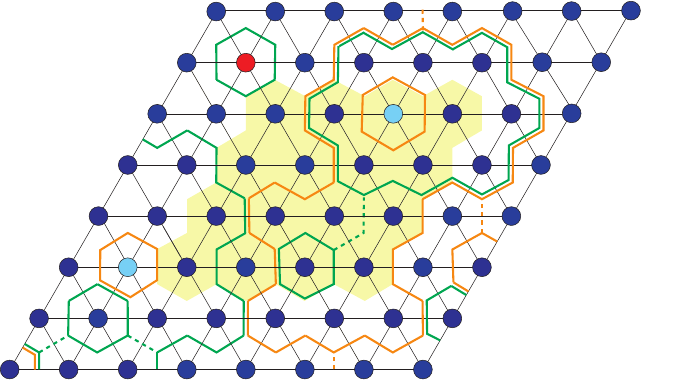}
		\captionsetup{justification=centering}
		\subcaption{}\label{3c}
	\end{subfigure}
	\caption[]{(a) A configuration $\sigma$ with  a $\delta-$race\underbar\ wealth bridge system. (b) Richness inversion in $\psi = f_1(\sigma)$. (c) Color inversion in $\tau = f_2(\psi)$. }\label{3}
\end{figure}

\subsection{The Inversion Mappings} \label{A:inv}
For any configuration $\sigma \in \Omega_{\rm t}$, after bridging, in order to get the upper bounds of $h(\nu) - h(\sigma)$ and $p(\sigma) - p(\nu)$, we define the mappings to first eliminate the bridged racially heterogeneous edges $h(\sigma)$ and bridged poor agents, and then define mappings to recover the racially heterogeneous edges and the poor agents on the urban sites in various designed ways. Different proof targets will lead to different recovery mappings, shown in the theorems in the following sections. For this part, we define the \textit{richness inversion} function $f_1(\cdot)$ and \textit{color inversion} function $f_2(\cdot)$ which eliminate most of the racially heterogeneous edges and poor agents respectively.

\vskip.1in
\noindent{\bf  Richness Inversion.} 
First, we represent each agent with two bits, with poor red denoted $00$, poor blue denoted $01$, rich red denoted $10$ and rich blue denoted $11.$ 
We define the  \textit{richness inversion} function $f_1(\cdot)$ as follows: for any agent $ij$, where $i$ is the richness bit, and $j$ is the color bit, the richness bit is flipped to ${(i + b)(\text{mod} 2)}$ for agent $ij$ that is surrounded by $b$ bridged richness contours or unbridged crossing richness contours (see the left corner's contour as an example of  richness contours in Figure \ref{3b}). The color bit remains unchanged. See Figure \ref{3b} for illustrations. 
\begin{lemma}[Lemma 7.5 in \cite{cannon2019local}] \label{I:f1}
	For any configuration $\psi = f_1(\sigma)$, there are at most $\delta n$ poor agents, and they are unbridged; no additional color edges are introduced; and for any mapped configuration $\psi$, there is only one preimage for a given $\delta-$richness bridge system.
\end{lemma}

\vskip.1in
\noindent{\bf Color Inversion.} 
We define the function \textit{color inversion} $f_2(\cdot)$ as: for any agent $ij$, where $i$ is the richness bit, and $j$ is the color bit, the color bit is flipped to ${(j + b)(\text{mod} 2)}$ for the agent $ij$ that is surrounded by $b$ bridged color contours or unbridged crossing color contours (like the red agent on the right boundary in Figure \ref{3b}). The richness bit remains unchanged. See Figure \ref{3c} for illustrations.
\begin{lemma}[Lemma 7.8 in \cite{cannon2019local}] \label{I:f2}
	 For any $\tau = f_2(\psi)$, $(x_{\rm c}+y_{\rm c})$ of the original racially heterogeneous edges in $\psi$ are eliminated: $h(\tau) - h(\psi) \leq - (x_{\rm c}+y_{\rm c})$;
	 no additional poor agents are introduced; and for any mapped configuration $\tau$ with a given bridge system, there is only one preimage that can be mapped to it.
\end{lemma}

\subsection{The Color and Richness Recovery Mappings} \label{c and r reovery}
After eliminating the bridged poor agents and the bridged racially heterogeneous edges in $f_1(\cdot)$ and $f_2(\cdot)$, we need to recover the same ratio of each color and richness as in $\sigma$, which is defined in $f_3(\cdot)$, $f_4(\cdot)$ and $f_5(\cdot)$ as the following.

\noindent{\bf Pink Recovery.}
For any $\tau = (f_2{\circ}f_1)(\sigma)$, we define the \textit{pink recovery} function $f_3(\tau)$ as to flip the agents' colors to pink starting from a fixed place in a given order except when encountering the following unbridged agents: we flip the unbridged pink to cyan, cyan to red, and red to blue. The flipping process stops once reaching the correct number of the pink agents as in $\sigma.$
 
\begin{lemma} \label{I:f3}
For any mapped configuration $\zeta$, if the starting location of $f_3$ and the flipping order are specified, there are at most $n$ preimages that can be mapped to it: $|f_3^{-1}(\zeta)| \leq n$. 
\end{lemma}
\begin{proof}
Given $\zeta$, it suffices to recover its preimage if we are given the stopping place and there are at most $n$ possible stopping places.
\end{proof}

\noindent{\bf Cyan Recovery.}
For any $\zeta = (f_3{\circ}f_2{\circ}f_1)(\sigma)$, we define the \textit{cyan recovery} function $f_4(\zeta)$ as to flip the agents starting from the stopping place of $f_3(\cdot)$ in a given order to cyan except when encountering the following unbridged agents: we remain the unbridged pink to pink, red to red, and flip cyan to blue. The flipping will be stopped after reaching the right number of the cyan agents as in $\sigma.$ The proof of Lemma \ref{I:f4} is similar to the proof of Lemma \ref{I:f3}. Given any mapped configuration $\phi$, it suffices to recover its preimage $f_4^{-1}(\phi)$ if we are given the starting location and the stopping place of $f_4(\cdot)$, which is bounded by $n^2.$
\begin{lemma}\label{I:f4}
For any mapped configuration $\phi$, if the flipping order is specified, there are at most $n^2$ preimages that can be mapped to it: $|f_4^{-1}(\phi)| \leq n^2$. 
\end{lemma}

\noindent{\bf Red Recovery.}
For any $\phi = (f_4{\circ}f_3{\circ}f_2{\circ}f_1)(\sigma)$, we define the \textit{red recovery} function $f_5(\cdot)$ as to flip agents starting from the stopping place of $f_4(\cdot)$ in a given order to red except when encountering the following unbridged agents: we flip the unbridged red to blue, cyan to pink, and pink to cyan.  To guarantee the right number of the cyan and pink in the mapped configuration $\nu$, whenever we flip an unbridged cyan to pink during this phase, we flip one pink back to cyan starting from the stop location of $f_3(\cdot)$. If we encounter non-pink agents, we first recover its colors before $f_3(\cdot)$ and use the rule of flipping the unbridged agents for $f_4(\cdot)$. Whenever we flip an unbridged pink to cyan, we flip one cyan to pink starting from the starting location of $f_4(\cdot)$, and if we encounter agents that are not cyan, we first recover its colors before $f_4(\cdot)$ and use the rule of $f_3(\cdot)$. We stop such operations after reaching the right number of the red for $\nu$. 

The proof of Lemma \ref{I:f5} is similar as Lemma \ref{I:f3}. To show the upper bound of the number of preimages $|f_5^{-1}(\nu)|$ for a given $\nu$: we first need to find the stop location of $f_4(\cdot)$, which has at most $n$ possibilities. Then we complement the colors of the first $k$ elements of $\nu$, where $k \in \{0, 1, ..., n-1\}$, and possibly we also need to find the stop location of $f_3$ (same location as the starting location of $f_4$), which has at most $n$ possibilities. Altogether there are at most $n^3$ different preimages.
\begin{lemma}\label{I:f5}
For any mapped configuration $\nu$, if the flipping order is specified, there are at most $n^3$ preimages that can be mapped to it: $|f_5^{-1}(\nu)| \leq n^3$. 
\end{lemma}

\subsection{The Centralized Recovery and Distributed Recovery} \label{recovery_way}
\begin{figure} 
	\begin{subfigure}[t]{0.3\textwidth} 
		\centering
		\includegraphics[height=1.3in]{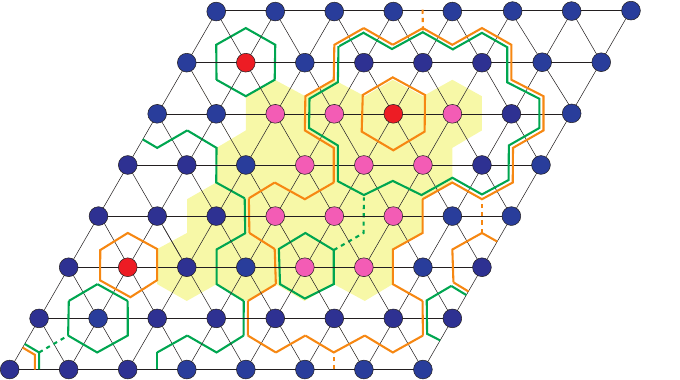}
		\captionsetup{justification=centering}
		\subcaption{}\label{4a}
	\end{subfigure}
	\begin{subfigure}[t]{0.3\textwidth} 
		\centering
		\includegraphics[height=1.3in]{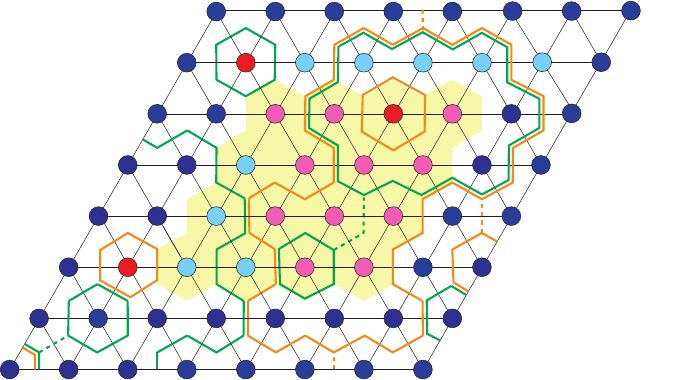}
		\captionsetup{justification=centering}
		\subcaption{}\label{4b}
	\end{subfigure}
	\begin{subfigure}[t]{0.3\textwidth} 
		\centering
		\includegraphics [height=1.3in]{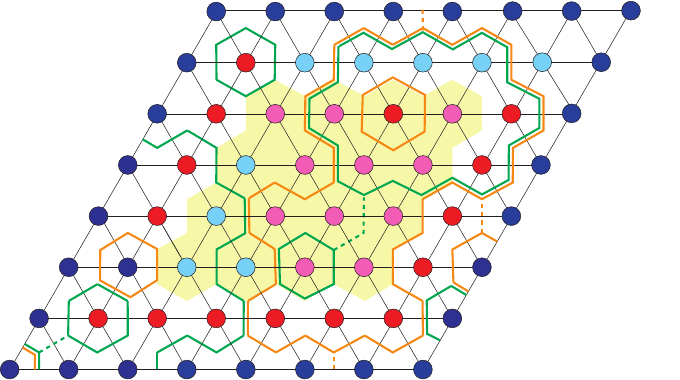}
		\captionsetup{justification=centering}
		\subcaption{}\label{4c}
	\end{subfigure}
	\caption[]{(a) Demonstration of $\zeta = f_3(\tau)$. (b) Demonstration of $\phi = f_4(\zeta)$. (c) Demonstration of $\nu = f_5(\phi)$. }\label{4}
\end{figure}

\noindent{\bf Centralized Recovery.}
If the urban sites are centralized like shown in Figure 1a in the paper, for the pink, cyan, and red recoveries defined in $f_3, f_4$ and $f_5$, the starting location of $f_3$ and the flipping order of each function can be specified in the \textit{centralized way}: the starting location of $f_3$ is specified to be the center of the urban area, and the flipping order for $f_3, f_4, f_5$ are specified as in clockwise direction and loop to the immediate outer layer when completing flipping one clockwise cycle like shown in Figure \ref{4a}, \ref{4b}, and \ref{4c}. In such a way, the following upper bound can be obtained.
\begin{lemma}\label{I:f5_h1}
	If the recoveries are specified in the centralized way, for any $\nu = (f_5{\circ}f_4{\circ}f_3)(\tau),$ there are at most $3\alpha \sqrt{n}$ racially heterogeneous edges introduced by the recovery operations: $h(\nu) - h(\tau) \leq 3\alpha \sqrt{n}$.
\end{lemma}
\begin{proof}
	In the defined way of centralized recovery, racially heterogeneous edges are created along the hexagon boundaries between the pink and cyan regions, the cyan and red regions, and the red and the rest region. Each boundary can be upper bounded by the perimeter of the fundamental domain, which is $\alpha \sqrt{n}$, and in total $3\alpha \sqrt{n}$.
	
	Inside each pink, cyan, and red region, the unbridged agents will not create additional racially heterogeneous edges: for any $\tau = (f_2{\circ}f_1)(\sigma)$, It follows from the definition of $f_2$ and Lemma \ref{I:f2} that the bridged agents in $\tau$ are either cyan or blue. The flipping rules for the unbridged agents in $f_3, f_4$ and $f_5$ thus can be verified to not introduce additional racially heterogeneous edges. Hence we get $h(\nu) - h(\tau) \leq 3\alpha \sqrt{n}.$ 
\end{proof}

\noindent{\bf Distributed Recovery.} If the urban sites evenly partition the city(diamond-shaped city) like shown in Figure 1b in the paper, for the pink, cyan, and red recoveries defined in $f_3, f_4$ and $f_5$, the starting location of $f_3$ and the flipping order of each function can be specified in the \textit{distributed way}: 
the starting location of $f_3$ is specified to be the most top left corner of the urban site, and the flipping order for $f_3, f_4, f_5$ are specified as to only flip the agents on the urban sites row after row, and then flip agents on the non-urban sites row after row after finishing flipping all the agents on the urban sites. In such a way, the following upper bound can be obtained.
\begin{lemma}\label{I:f5_h2}
	If the recoveries are specified in the distributed way, for any $\nu = (f_5{\circ}f_4{\circ}f_3)(\tau),$ there are at most $2 c n$ racially heterogeneous edges introduced by the recovery operations: $h(\nu) - h(\tau) \leq 2cn$.
\end{lemma}
\begin{proof}
		In the defined way of distributed recovery, compared with $\tau$, racially heterogeneous edges are possibly created along the diamond boundaries between the urban and non-urban sites (see Figure 1b in the paper for demonstration), which can be upper bounded by the sum of perimeters of all the small diamonds. 
		
		The perimeter of a diamond can be upper bounded by $4 \cdot k$:  since the urban sites, with total size $c \cdot n$, evenly partitions the finite lattice, the number of row urban sites is $\frac{cn}{2}.$ Each row of urban sites has $\sqrt{n}$ sites, so there are $\frac{c\sqrt{n}}{2}$ rows of urban sites. Since each column of urban sites also has $\sqrt{n}$ sites, we can get the side length of a diamond to be $k = \frac{\sqrt{n}}{\frac{c\sqrt{n}}{2} -1} \geq \frac{2}{c}$. Since the side length of one diamond is $k,$ the total number of the diamonds can be upper bounded by $(\frac{\sqrt{n}}{k})^2$. Hence the sum of perimeters of all the diamonds can be upper bounded by $\frac{4n}{k} \leq 2cn,$ which yields $h(\nu) - h(\tau) \leq 2cn.$ 
\end{proof}

\begin{lemma}\label{I:f5_p}
	For any $\nu = (f_5{\circ}f_4{\circ}f_3{\circ}f_2{\circ}f_1)(\sigma),$ if the recovery is specified in the centralized way or distributed way, the number of the poor agents on the urban sites is lower bounded by $p(\nu) \geq (\min\{c, p\} - \delta)n.$
\end{lemma}
\begin{proof}
	It follows from Lemma \ref{I:f1} that the number of  poor agents is at most $\delta n$ in the mapped $\psi = f_1(\sigma)$ and they are unbridged. It follows from Lemma \ref{I:f2} that no additional poor agents will be introduced for any $\tau = f_2(\psi).$ In the worst case scenario, these $\delta n$ poor agents will not be recovered on the urban sites in $\nu = (f_{5} {\circ} f_{4} {\circ} f_{3})(\tau)$. For the rest of the poor agents, because of the recovery way in the order of pink followed by cyan then red, it yields that $p(\nu) \geq  (\min\{c, p\} - \delta)n$. 
\end{proof}

\section{Proof Supports of Theorem 2 (Urbanization of Poverty)} \label{A:urb1}

\begin{claim}\label{C:1}
	For any $\sigma \in \Omega_{\neg \text{urb}}$ with a given bridged color contour length $z_{\rm c}$, for the defined mapping $\nu = f(\sigma),$ $h(\nu) - h(\sigma) \leq 3\alpha \sqrt{n}  - z_{\rm c}$ and $p(\sigma) - p(\nu) \leq - \delta n$, where $\delta = \epsilon /2.$
\end{claim}
\begin{proof}[Proof of Claim \ref{C:1}]
	It follows from Lemma \ref{I:f2} that for any $\sigma \in \Omega_{\neg \text{urb}}$, $h(\tau) - h(\sigma) \leq - z_{\rm c}$. 
	It follows from Lemma \ref{I:f5_h1} that $h(\nu) - h(\tau) \leq 3\alpha \sqrt{n}.$ Combining the two inequalities, we get $h(\nu) - h(\sigma) \leq 3\alpha \sqrt{n}  - z_{\rm c}$.
	
	For any $\sigma \in \Omega_{\neg \text{urb}}$, $p(\sigma) \leq (\min\{c, p\} - \epsilon)n$ is satisfied. It follows from Lemma \ref{I:f5_p} that $p(\nu) \geq  (\min\{c, p\} - \delta)n$. Combining the inequalities, setting $\epsilon = 2\delta$, it follows that $p(\sigma) - p(\nu) \leq -\epsilon n + \delta n = -\delta n$.	
\end{proof}

\begin{claim}\label{C:2}
	For a given color contour length $z_{\rm c}$, for any $\nu = f(\sigma)$, the number of configurations in $ \Omega_{\neg \text{urb}}$ that can map to $\nu$ is upper bounded by: $|f^{-1}(\nu)| \leq   (z_{\rm c}+1) 9^{\alpha \sqrt{n}}4^{(\frac{3\delta + 1}{4\delta})(z_{\rm c}+3n)} .$ 
\end{claim}
\begin{proof}[Proof of Claim \ref{C:2}]
	We denote $\nu = f_5(\phi)$, $\phi = f_4(\zeta)$, $\zeta = f_3(\tau)$, $\tau = f_2(\psi)$, and $\psi = f_1(\sigma)$. 
	Since $f(\sigma) = (f_5{\circ}f_4{\circ}f_3{\circ}f_2{\circ}f_1)(\sigma)$, 
	for any $\nu$, it follows that $|f^{-1}(\nu)| \leq |f_5^{-1}(\nu)| \cdot |f_4^{-1}(\phi)| \cdot  |f_3^{-1}(\zeta)| \cdot |f_2^{-1}(\tau)| \cdot |f_1^{-1}(\psi)|$.
	It follows from Lemma \ref{I:f5} that $|f_5^{-1}(\nu)| \leq n^3$. It follows from Lemma \ref{I:f4} that $|f_4^{-1}(\phi)| \leq n^2$. It follows from Lemma \ref{I:f3} that $|f_3^{-1}(\zeta)| \leq n$.

	For a given color contour length $z_{\rm c}$, the number of all possible bridge systems $b$ is upper bounded by $b \leq (\max\{z_{\rm r} \} + 1)(z_{\rm c} + 1) 9^{\alpha \sqrt{n}}4^{(\frac{3\delta +1}{4\delta})(z_{\rm c} + \max\{z_{\rm r} \})}$ in Lemma \ref{I:bridge_num}. 
	For the triangular lattice with $n$ vertices, the sum of all edges is $3n$. Because of the lattice duality, we conclude $z_{\rm r} \leq 3n.$ 
	Thus $b \leq (3n + 1)(z_{\rm c} + 1) 9^{\alpha \sqrt{n}}4^{(\frac{3\delta +1}{4\delta})(z_{\rm c} + 3n)}$.
	It follows from Lemma \ref{I:f1} and \ref{I:f2} that for a given bridge system and a given image, the number of the corresponding preimages of both $f_2$ and $f_1$ is one. 
	
	Hence, we conclude that for a given $z_{\rm c}$, $|f^{-1}(\nu)| \leq |f_5^{-1}(\nu)| \cdot |f_4^{-1}(\phi)| \cdot  |f_3^{-1}(\zeta)| \cdot b \leq n^6 (3n + 1)(z_{\rm c} + 1) 9^{\alpha \sqrt{n}}4^{(\frac{3\delta +1}{4\delta})(z_{\rm c} + 3n)}$.	
\end{proof}

\section{Proof of Corollary  3 (Distributed Urbanization of Poverty)} \label{A:colurb}
\begin{proof}[Proof of Corollary 3]
	The bridge system and mapping are the same as the proof of Theorem 2 except that: for $(f_{5} {\circ} f_{4} {\circ} f_{3}) (\tau)$, we recover the same ratios of each color and richness as in $\sigma$ in the distributed way defined in section \ref{recovery_way}. Claim \ref{C:2} still holds true. For Claim \ref{C:1}, it now follows from Lemma \ref{I:f5_h2} and \ref{I:f2} that $h(\nu) -h(\sigma) \leq 2cn - z_{\rm c} \leq 3n - z_{\rm c}$, and $p(\sigma) -p(\nu) \leq -\delta n$ holds true, where $\delta = \epsilon /2.$ Substituting the bounds into the Peierls argument yields
	\begin{align}\label{urbproof2}
		\pi(\Omega_{\neg \text{urb}}) &\leq \sum_{\nu \in \Omega} \pi(\nu) \sum_{z_{\rm c} = \sqrt{r\cdot n}}^{3n} n^{6}(3n + 1)(z_{\rm c}+1) 9^{\alpha\sqrt{n}}(\frac{4^{\frac{3\delta + 1}{4\delta}}}{\lambda})^{z_{\rm c}} (\frac{\lambda^3 64^{\frac{3\delta + 1}{4\delta}}}{\gamma^{\delta}})^{n}.
	\end{align}
	Similarly, when $\lambda \geq 4^{\frac{3\delta + 1}{4\delta}}$, if $\gamma^{\delta/3} >  \lambda \cdot 4^{\frac{3\delta + 1}{4\delta}}$, the sum will be exponentially small for sufficiently large $n$. When $ 1 \leq \lambda < 4^{\frac{3\delta + 1}{4\delta}}$, the sum further yields $\pi(\Omega_{\neg \text{urb}})  \leq n^{6} \cdot (3n +1) \cdot 9^{\alpha\sqrt{n}} \cdot 3n \cdot (\frac{\lambda \cdot 16^{\frac{3\delta + 1}{4\delta}}}{\lambda\gamma^{\delta/3}})^{3n}.$ As long as $\gamma^{\delta/3} > 16^{\frac{3\delta + 1}{4\delta}}$, the sum will still be exponentially small for sufficiently large $n$. Combining the two cases, we can see that as long as $\gamma^{\delta/3} > 4^{\frac{3\delta + 1}{4\delta}} \max\{\lambda, 4^{\frac{3\delta + 1}{4\delta}} \}$ and $\lambda > 1$, $\pi(\Omega_{\neg \text{urb}}) \leq \xi_{1}^{n},$ for $\xi_{1} \in (0,1).$ Substituting $\delta = \frac{\epsilon}{2}$ the lower bound of $\gamma$ yields Corollary 3.	
\end{proof}


\section{Proof of Theorem 4 (Dispersion of Poverty)} \label{deurb:D}
\begin{proof}[Proof of Theorem 4]
	Let $\Omega_{\rm{urb}} \in \Omega$ be the set of configurations that have $\epsilon-$urbanization of poverty. To prove Theorem 4, it suffices to prove $\pi(\Omega_{\rm{urb}}) \leq \xi_{2}^{n}$ for some $\xi_2 \in (0, 1)$ and large enough $n$.
	
	For each $\sigma \in \Omega_{\rm{urb}},$ we construct a $\delta-$richness bridge system and define the mapping $\nu = d(\sigma) = (d_2 \circ f_1)(\sigma)$ as the following:
	we first do the richness inversion and obtain $\tau = f_1(\sigma)$; next for $\tau$, we randomly flip the blue to cyan until the right number of the cyan, and we randomly flip the red to pink until the right number of the pink and obtain $\nu = d_2(\tau).$

	\begin{claim}\label{C:6}
		For any $\sigma \in \Omega_{\rm{urb}}$, for the defined mapping $\nu = d(\sigma),$ $h(\nu) - h(\sigma) \leq 0$ and $p(\sigma) - p(\nu) \leq cn.$
	\end{claim}
\begin{proof}[Proof of Claim \ref{C:6}]
	Because the color information remains unchanged for $f_1$ and $d_2$, $h(\nu) - h(\sigma) \leq 0$. For any $\sigma$, $p(\sigma) \leq \min\{c, p\}n$. The scenario under discussion follows that $p < c$, thus $p(\sigma) \leq pn.$ Since for any $\nu$, $p(\nu) \geq 0,$ hence $p(\sigma) - p(\nu) \leq pn.$
\end{proof}

	\begin{claim}\label{C:7}
		For a given richness contour length $z_{\rm r}$, for any $\nu = d(\sigma)$, the number of configurations in $ \Omega_{\rm{urb}}$ that can map to $\nu$ is upper bounded by: $|d^{-1}(\nu)| \leq (\beta \sqrt{n} +1)3^{\alpha \sqrt{n}}4^{\frac{3\delta + 1}{4\delta}\beta \sqrt{n}}2^{pn}.$ 
	\end{claim}
\begin{proof}[Proof of Claim \ref{C:7}]
	We denote $\tau = f_1(\sigma)$ and $\nu = d_2(\tau)$. Since $d(\sigma) = (d_2 {\circ} f_1)(\sigma)$, for any $\nu,$ it follows that $|d^{-1}(\nu)| \leq |f_{1}^{-1}(\tau)| \cdot |d_{2}^{-1}(\nu)|$.  For any $\nu \in \Omega,$ the number of preimages $|d_2^{-1}(\nu)|$ can be upper bounded by $2^{pn}$ by recording whether each poor agent is flipped or not. 
	
	Because every configuration $\sigma \in \Omega_{\rm{urb}}$ satisfies $\epsilon-$urbanization, assuming we also have $p < c < p+\epsilon,$ then $\sigma$ also satisfies $(\beta, \epsilon)-$wealth segregation (similarly defined as Definition 5 in the paper). Hence $z_{\rm r}$ can be upper bounded by $z_{\rm r} \leq \beta \sqrt{n}$ (See Lemma 7.4 in \cite{cannon2019local} for details). Thus it follows from Lemma \ref{I:rbridge} that for any given $z_{\rm r}$, the number of $\delta-$richness bridge systems can be upper bounded by $(\beta \sqrt{n} +1)3^{\alpha \sqrt{n}}4^{\frac{3\delta + 1}{4\delta}\beta \sqrt{n}}$. It follows from Lemma \ref{I:f1} that for any $\tau$, for a given $\delta-$richness bridge system, the number of preimages is one. Hence we conclude for any $\tau$ with a given $z_{\rm r}$, the number of preimages is upper bounded by: $|f_1^{-1}(\tau)| \leq (\beta \sqrt{n} +1)3^{\alpha \sqrt{n}}4^{\frac{3\delta + 1}{4\delta}\beta \sqrt{n}}$. Combining the two inequalities, we conclude $|d^{-1}(\nu)| \leq (\beta \sqrt{n} +1)3^{\alpha \sqrt{n}}4^{\frac{3\delta + 1}{4\delta}\beta \sqrt{n}} 2^{pn}$.
\end{proof}

	\begin{claim}\label{C:8}
		For any given $\sigma \in \Omega_{\rm{urb}},$ we denote $D(\sigma)$ to be the set of all possible images $\nu = d(\sigma)$ mapped from $\sigma$. It follows that $|D(\sigma)| \geq (\frac{r-\delta}{r_{\rm p}})^{(r_{\rm p}-\delta)n} (\frac{b-\delta}{b_{\rm p}})^{(b_{\rm p}-\delta)n}$.
	\end{claim}
\begin{proof}[Proof of Claim \ref{C:8}]
	It follows from the definition of $f_1$ that for any given $\sigma,$ there is only one configuration can be obtained from $f_1(\sigma)$.
	For any given $\tau$, we define $D_2(\tau)$ to be the set of all possible configurations obtained from $d_2(\tau)$ by flipping the pink and cyan back; then $|D_2(\tau)| = \begin{pmatrix}
		rn - a_{\tau}\\ 
		r_{\rm p}n - a_{\tau}^{'}
	\end{pmatrix} \cdot \begin{pmatrix}
		bn - a_{\tau}\\ 
		b_{\rm p}n - a_{\tau}^{''}
	\end{pmatrix}$, where $a_{\tau}$ is the number of unbridged agents in $\tau$, $a_{\tau}^{'}$ is the number of unbridged pink, and $a_{\tau}^{''}$ is the number of unbridged cyan. Thus $|D_2(\tau)|$ can be further lower bounded by 
	\begin{align}
		|D_2(\tau)|\geq (\frac{r - a_{\tau}}{r_{\rm p} - a_{\tau}^{'}})^{(r_{\rm p} - a_{\tau}^{'})n} (\frac{b - a_{\tau}}{b_{\rm p} - a_{\tau}^{''}})^{(b_{\rm p} - a_{\tau}^{''})n} \geq (\frac{r-\delta}{r_{\rm p}})^{(r_{\rm p}-\delta)n} (\frac{b-\delta}{b_{\rm p}})^{(b_{\rm p}-\delta)n}.
	\end{align}
	Hence $|D(\sigma)|$ can be lower bounded by $|D(\sigma)| \geq 1 \cdot |D_2(\tau)| \geq (\frac{r-\delta}{r_{\rm p}})^{(r_{\rm p}-\delta)n} (\frac{b-\delta}{b_{\rm p}})^{(b_{\rm p}-\delta)n}.$
\end{proof}

	Finally, we define a weighted bipartite graph $G(\Omega_{\rm{urb}}, \Omega,E)$ with an edge of weight $\pi(\sigma)$ between $\sigma \in \Omega_{\rm{urb}}$ and $\nu \in \Omega$. The total weight of edges is
	\begin{align}\label{LHS}
		\sum_{\sigma \in \Omega_{\rm{urb}}} \pi(\sigma) \cdot |D(\sigma)| \geq \pi(\Omega_{\rm{urb}})  (\frac{r-\delta}{r_{\rm p}})^{(r_{\rm p}-\delta)n} (\frac{b-\delta}{b_{\rm p}})^{(b_{\rm p}-\delta)n}. 	
	\end{align}
	On the other hand, the weight of the edges is at most 
	\begin{align} \label{RHS}
		&\sum_{\nu \in \Omega} \sum_{\sigma \in d^{-1}(\nu)}\max_{\sigma \in \Omega_{\rm{urb}}}{\pi(\sigma)} =  \sum_{\nu \in \Omega} \pi(\nu) \sum_{\sigma \in d^{-1}(\nu)} \frac{\max_{\sigma \in \Omega_{\rm{urb}}}{\pi(\sigma)}}{\pi(\nu)}  |d^{-1}(\nu)|\nonumber\\
		&\leq \sum_{\nu \in \Omega} \pi(\nu) \sum_{z_{\rm r = \sqrt{pn}}}^{\beta \sqrt{n}} \lambda^{\max(h(\nu)-h(\sigma))} \gamma^{\max(p(\sigma)-p(\nu))} (\beta \sqrt{n} +1)3^{\alpha \sqrt{n}}4^{\frac{3\delta + 1}{4\delta}\beta \sqrt{n}} 2^{pn}
		\nonumber \\&\leq \gamma^{pn}  \beta \sqrt{n} (\beta \sqrt{n} +1)3^{\alpha \sqrt{n}}4^{\frac{3\delta + 1}{4\delta}\beta \sqrt{n}} 2^{pn},
	\end{align}
	where the inequalities in Claim \ref{C:6} and \ref{C:7} has been substituted in the above derivation. Combining \eqref{LHS} and \eqref{RHS},  we have
	\begin{align*}
		\pi(\Omega_{\rm{urb}})  (\frac{r-\delta}{r_{\rm p}})^{(r_{\rm p}-\delta)n} (\frac{b-\delta}{b_{\rm p}})^{(b_{\rm p}-\delta)n} \leq \beta \sqrt{n} (\beta \sqrt{n} +1)3^{\alpha \sqrt{n}}4^{\frac{3\delta + 1}{4\delta}\beta \sqrt{n}} \gamma^{cn} 2^{pn}.
	\end{align*}
	For large enough $n$, to have $\pi(\Omega_{\rm{urb}}) \leq \xi_3^{n}$ for some $\xi_3 \in (0, 1)$, it suffices to have $$\gamma^{pn}2^{pn} < (\frac{r-\delta}{r_{\rm p}})^{(r_{\rm p}-\delta)n} (\frac{b-\delta}{b_{\rm p}})^{(b_{\rm p}-\delta)n} < (\frac{r-\delta}{r_{\rm p}})^{r_{\rm p}n} (\frac{b-\delta}{b_{\rm p}})^{b_{\rm p}n},$$ which can be rewritten as 
	$$\gamma < (\frac{r-\delta}{r_{\rm p}})^{\frac{r_{\rm p}}{p}} (\frac{b-\delta}{b_{\rm p}})^{\frac{b_{\rm p}}{p}} / 2.$$
	Since $\gamma > 1,$ to make the right hand side of the above inequality greater than one, it suffices to have $\frac{r-\delta}{r_{\rm p}} > 2$ and $\frac{b-\delta}{b_{\rm p}} > 2$, which can be rewritten as $r_{\rm p} < r_{\rm r} - \delta$ and $b_{\rm p} < b_{\rm r} - \delta.$  
\end{proof}


\section{Proof Supports of Theorem 7 (Urbanized Racial Segregation)} \label{A:red urb}

	\begin{claim}\label{C:12}
	If $\lambda > 3^{\frac{\alpha}{\beta}} 4^{\frac{3\delta+1}{4\delta}}$, $\pi(\Omega_{\text{urb} \wedge \neg \text{seg}}) \leq \xi_{0}^{\sqrt{n}}$ for some $\xi_0 \in (0, 1).$
\end{claim}
\begin{proof}[Proof of Claim \ref{C:12}]
	To prove $\pi(\Omega_{\text{urb} \wedge \neg \text{seg}}) \leq \xi_{0}^{\sqrt{n}}$, we define the bridge system and the mapping $g(\cdot)= (g_{3} {\circ} g_{2} {\circ} f_{2}) (\cdot)$ from the set $\Omega_{\text{urb} \wedge \neg \text{seg}}$ to $\Omega$ as the following: 	for any $\sigma \in \Omega_{\text{urb} \wedge \neg \text{seg}}$, we construct a $\delta-$color bridge system for it. Then we do color inversion like defined in $f_2(\cdot)$, and get $\psi = f_2(\sigma).$ 
	
	$\tau = g_2(\psi)$ is defined as the following: starting from the center of the urban area, we flip the color bit ($0$ to $1$ and $1$ to $0$) of each agent layer by layer in a given order until the right number of pink and cyan is reached.
	
	$\nu = g_3(\tau)$ is then defined as the following: starting from outside the urban area boundary, we flip the color bit ($0$ to $1$ and $1$ to $0$) of each agent layer by layer in a given order until the right number of red and blue is reached. During this phase, whenever we flip the color bit for an unbridged pink, we go back to the stopping location of $g_2(\cdot)$ to flip one more cyan to pink, and vise versa.

	\begin{claim}\label{C:4}
		For any $\sigma \in \Omega_{\text{urb} \wedge \neg \text{seg}}$ with a given bridged color contour length $z_{\rm c}$, for the defined mapping $\nu = g(\sigma),$ $h(\nu) - h(\sigma) \leq 2\alpha \sqrt{n}  - z_{\rm c}$ and $p(\sigma) - p(\nu) \leq 0.$
	\end{claim}
	\begin{proof}[Proof of Claim \ref{C:4}]
		It follows from Lemma \ref{I:f2} that $h(\psi) - h(\sigma) \leq -z_{\rm c}$ and the richness information remains unchanged. For the color recoveries $(g_3 {\circ} g_2)(\psi)$, similar as the proving strategy of Lemma \ref{I:f5_h1}, we at most introduce $2\alpha \sqrt{n}$ racially heterogeneous edges compared with $h(\psi)$ due to the hexagon boundaries created between pink and cyan, and cyan and red. Each boundary can be upper bounded by the perimeter of the fundamental domain $\alpha \sqrt{n}$.  The richness information remains unchanged. Thus $h(\nu) - h(\sigma) \leq h(\nu) - h(\psi) + h(\psi) -h(\sigma) \leq 2\alpha \sqrt{n}  - z_{\rm c}$ and $p(\sigma) - p(\nu) \leq 0.$
	\end{proof}

	\begin{claim}\label{C:5}
		For a given color contour length $z_{\rm c}$, for any $\nu = g(\sigma)$, the number of configurations in $ \Omega_{\text{urb} \wedge \neg \text{seg}}$ that can map to $\nu$ is upper bounded by: $|g^{-1}(\nu)| \leq n^3 (z_{\rm c}+1)3^{\alpha \sqrt{n}}4^{\frac{3\delta + 1}{4\delta}z_{\rm c}}.$ 
	\end{claim}
	\begin{proof}[Proof of Claim \ref{C:5}]
		For any given mapped $\tau$, the number of preimages $|g_2^{-1}(\tau)|$ is at most $n$, since we only need to know the stopping locations of the flipping operations, and there are at most $n$ possibilities of the stopping location. For any given mapped $\nu$, the number of preimages $|g_3^{-1}(\tau)|$ is at most $n^2$, since we need the stopping location information of $g_2(\cdot),$ which is upper bounded by $n$ and the stopping location of $g_3(\cdot)$ is also upper bounded by $n.$
		
		To bound $|g^{-1}(\nu)|$ for a given $\nu$ and a given $z_{\rm c} = x_{\rm c} + y_{\rm c}$, we can first bound the number of possible bridge systems for a given $z_{\rm c}$, which yields $(z_{\rm c}+1)3^{\alpha \sqrt{n}}4^{\frac{3\delta + 1}{4\delta}z_{\rm c}}$. See proof details of this bound from Lemma 7.6 in \cite{cannon2019local}. 
		For any configuration $\psi \in f_2(\sigma)$ with a given bridge system, there is only one configuration $\sigma$ that can be mapped to $\psi$ (Lemma \ref{I:f2}). 
		Thus combining with $|g_2^{-1}(\tau)|$ and $|g_3^{-1}(\nu)|$, it yields $|g^{-1}(\nu)| \leq n^3 (z_{\rm c}+1)3^{\alpha \sqrt{n}}4^{\frac{3\delta + 1}{4\delta}z_{\rm c}}.$	
	\end{proof}

	Finally, substituting the bounds into Peierls Argument of Equation (1) of the paper yields
	\begin{align}\label{urbproof1}
		\pi(\Omega_{\text{urb} \wedge \neg \text{seg}}) &\leq \sum_{\nu \in \Omega} \pi(\nu) \sum_{z_{\rm c} = \beta \sqrt{n}}^{3n} n^3 (z_{\rm c}+1) 3^{\alpha\sqrt{n}} (\frac{4^{\frac{3\delta + 1}{4\delta}}}{\lambda})^{z_{\rm c}} \nonumber \\
		&\leq \sum_{\nu \in \Omega} \pi(\nu) \sum_{z_{\rm c} = \beta \sqrt{n}}^{3n} n^3 (z_{\rm c}+1) (\frac{3^{\frac{\alpha}{\beta}} 4^{\frac{3\delta+1}{4\delta}}}{\lambda})^{z_{\rm c}},
	\end{align}				
	where  $z_{\rm c} \geq \beta\sqrt{n}$ is due to $\sigma \in \Omega_{\text{urb} \wedge \neg \text{seg}}$ does not satisfy $(\beta,\epsilon)-$segregation (see Lemma 7.4 in \cite{cannon2019local} for details), and $z_{\rm c} \leq \beta\sqrt{n}$, which is the sum of all edges of $G_{\triangle}$.
	If $\lambda > 3^{\frac{\alpha}{\beta}} 4^{\frac{3\delta+1}{4\delta}}$, the sum will be exponentially small for sufficiently large $n$, which means $\pi(\Omega_{\text{urb} \wedge \neg \text{seg}}) \leq \xi_{0}^{\sqrt{n}}$ for some $\xi_0 \in (0,1).$ 
\end{proof}


\section{Proof Supports of Theorem 8 (Urbanization and Integration)} \label{A2}
	\begin{claim}\label{C:9}
		For any $\sigma \in \Omega_{\text{urb} \wedge \text{seg}}$ with bridged color contour length $z_{\rm c}$, for the defined mapping $\nu = s(\sigma),$ $h(\nu) - h(\sigma) \leq -z_{\rm c} + 3n$ and $p(\sigma) - p(\nu) \leq 0.$
	\end{claim}
	\begin{proof}[Proof of \ref{C:9}]
		Because the richness information remains unchanged for $f_1$ and $s_2$, $p(\sigma) - h(\nu) \leq 0$. It follows from Lemma \ref{I:f2} that for a given $z_{\rm c}$ and any $\tau = f_2(\sigma)$, $h(\tau) -h(\sigma) \leq -z_{\rm c}$. The maximal number of racially heterogeneous edges created by $\nu = s_2(\tau)$ can be bounded by $3n$, which is sum of all the edges in $G_{\triangle}$. Hence $h(\nu) - h(\sigma) = h(\nu) - h(\tau) +  h(\tau) - h(\sigma) \leq 3n - z_{\rm c}.$
	\end{proof}
	
	\begin{claim}\label{C:10}
		For a given color contour length $z_{\rm c}$, for any $\nu = s(\sigma)$, the number of configurations in $ \Omega_{\text{urb} \wedge \text{seg}}$ that can map to $\nu$ is upper bounded by: $|s^{-1}(\nu)| \leq (z_{\rm c} +1)3^{\alpha \sqrt{n}}4^{\frac{3\delta + 1}{4\delta} z_{\rm c}}2^{rn}.$ 
	\end{claim}
	\begin{proof}[Proof of Claim \ref{C:10}]
		We denote $\tau = f_2(\sigma)$ and $\nu = s_2(\tau)$. Since $s(\sigma) = (s_2 {\circ} f_2)(\sigma)$, for any $\nu,$ it follows that $|s^{-1}(\nu)| \leq |f_{2}^{-1}(\tau)| \cdot |s_{2}^{-1}(\nu)|$.  For any $\nu \in \Omega,$ the number of preimages $|s_2^{-1}(\nu)|$ can be upper bounded by $2^{rn}$ by recording whether each red agent is flipped or not. 
		
		It follows from Lemma \ref{I:cbridge} that for any given $z_{\rm c}$, the number of $\delta-$color bridge systems can be upper bounded by $(z_{\rm c} +1)3^{\alpha \sqrt{n}}4^{\frac{3\delta + 1}{4\delta}z_{\rm c}}$. It follows from Lemma \ref{I:f2} that for any $\tau$, for a given $\delta-$color bridge system, the number of preimages is one. Hence we conclude for any $\tau$ with a given $z_{\rm c}$, the number of preimages is upper bounded by: $|f_2^{-1}(\tau)| \leq (z_{\rm c} +1)3^{\alpha \sqrt{n}}4^{\frac{3\delta + 1}{4\delta}z_{\rm c}}$. Combining the two inequalities, we conclude $|s^{-1}(\nu)| \leq (z_{\rm c} +1)3^{\alpha \sqrt{n}}4^{\frac{3\delta + 1}{4\delta}z_{\rm c}} 2^{rn}$.
	\end{proof}
	\begin{claim}\label{C:11}
		For any given $\sigma \in  \Omega_{\text{urb} \wedge \text{seg}},$ we denote $S(\sigma)$ to be the set of all possible images $\nu = s(\sigma)$ mapped from $\sigma$. It follows that $|S(\sigma)| \geq (\frac{p-\delta}{r_{\rm p}})^{(r_{\rm p}-\delta)n} (\frac{1-p-\delta}{r_{\rm r}})^{(r_{\rm r}-\delta)n}.$
	\end{claim}
	\begin{proof}[Proof of Claim \ref{C:11}]
		It follows from the definition of $f_2$ that for any given $\sigma,$ there is only one configuration can be obtained from $f_2(\sigma)$.
		For any given $\tau$, we define $S_2(\tau)$ to be the set of all possible configurations obtained from $s_2(\tau)$ by flipping the pink and red back; then $|S_2(\tau)| = \begin{pmatrix}
			pn - a_{\tau}\\ 
			r_{\rm p}n - a_{\tau}^{'}
		\end{pmatrix} \cdot \begin{pmatrix}
			n - pn - a_{\tau}\\ 
			r_{\rm r}n - a_{\tau}^{''}
		\end{pmatrix}$, where $a_{\tau}$ is the number of unbridged agents in $\tau$, $a_{\tau}^{'}$ is the number of unbridged pink, and $a_{\tau}^{''}$ is the number of unbridged red. Thus $|S_2(\tau)|$ can be further lower bounded by 
		\begin{align}
			|S(\tau)|\geq (\frac{p - a_{\tau}}{r_{\rm p} - a_{\tau}^{'}})^{(r_{\rm p} - a_{\tau}^{'})n} (\frac{1 -p - a_{\tau}}{r_{\rm r} - a_{\tau}^{''}})^{(r_{\rm r} - a_{\tau}^{''})n} \geq (\frac{p-\delta}{r_{\rm p}})^{(r_{\rm p}-\delta)n} (\frac{1-p-\delta}{r_{\rm r}})^{(r_{\rm r}-\delta)n}.
		\end{align}
		Hence $|S(\sigma)|$ can be lower bounded by $|S(\sigma)| \geq 1 \cdot |S_2(\tau)| \geq (\frac{p-\delta}{r_{\rm p}})^{(r_{\rm p}-\delta)n} (\frac{1-p-\delta}{r_{\rm r}})^{(r_{\rm r}-\delta)n}.$
	\end{proof}




\section{Proof Support of Theorem 10 (Integration for Distributed $\mathcal{U}$)} \label{A:miti}
	\begin{claim}\label{C:3}
		For any $\sigma \in S_{\beta, \delta}$ with a given bridged color contour length $z_{\rm c}$, for the defined mapping $\nu = f(\sigma),$ $h(\nu) - h(\sigma) \leq 2cn  - z_{\rm c}$ and $p(\sigma) - p(\nu) \leq b_{\rm p}n - \hat{b}_{\rm p}n$, where $\hat{b}_{\rm p} \triangleq \min\{c, p\}-(r+\delta)c - 2\delta.$
	\end{claim}
\begin{proof}[Proof of Claim \ref{C:3}]
	For any configuration $\sigma \in S_{\beta, \delta}$, the poor agents on the urban sites are either in the red region $R$ or outside $R$. It follows from the definition of $(\beta, \delta)-$segregation that the size of $R$ is at most $(r + \delta)n$.
	Since the urban sites are evenly distributed with the total size $c \cdot n$, the maximal number of the urban sites in $R$ is $(r + \delta) cn.$ Hence the maximal number of the poor agents on the urban sites in $R$ is $(r + \delta) cn.$ Outside $R$, the poor agents on the urban sites could be unbridged agents whose number is upper bounded by $\delta \cdot n$, or the cyan agents whose number is $b_{\rm p}n$. Hence the total number of the poor on the urban sites for any $\sigma$ follows that $p(\sigma) \leq (r + \delta) cn + \delta n + b_{\rm p}n.$

	After the richness and color inversions, it follows from Lemma \ref{I:f1} and \ref{I:f2} that for any $\tau = (f_2\circ f_1)(\sigma),$ $h(\tau) - h(\sigma) \leq -z_{\rm c}$ and the number of the poor in $\tau$ is less than $\delta n$ and they are unbridged.
	After the color and richness recovery, it  follows from Lemma \ref{I:f5_p} that $p(\nu) \geq \min\{c, p\}n - \delta n$. Hence $p(\sigma) - p(\nu) \leq  b_{\rm p}n - \hat{b}_{\rm p} n$, where $\hat{b}_{\rm p} = \min\{c, p\}-(r+\delta)c - 2\delta$. 
	
	It also follows from Lemma \ref{I:f5_h2} that for any $\nu = (f_5{\circ}f_4{\circ}f_3{\circ})(\tau)$, $h(\nu) - h(\tau) \leq 2cn.$
	Combining with $h(\tau) - h(\sigma) \leq -z_{\rm c}$, we get $h(\nu) - h(\sigma) \leq 2cn - z_{\rm c}.$ 
\end{proof}


\end{document}